\documentclass[conference]{IEEEtran}

\usepackage{amsmath,amsfonts}
\usepackage{graphicx}
\usepackage{textcomp}
\usepackage{color}
\usepackage{colortbl}
\usepackage{multirow}
\usepackage[T1]{fontenc}
\usepackage[utf8]{inputenc}
\usepackage{xspace}
\usepackage{latexsym} 
\usepackage[usenames,dvipsnames]{pstricks}
\usepackage{epsfig}
\usepackage{pst-plot}

\usepackage{array}            
\usepackage{multicol}         
\usepackage{graphicx}         
\usepackage{latexsym}         
\usepackage{amsthm}
\usepackage{systeme}

\usepackage{url}

\usepackage{subfig}
\usepackage{epstopdf}
\usepackage{arydshln}
\usepackage{pgf,tikz,pgfplots}
    
\usepackage{algorithm}
\usepackage[noend]{algpseudocode}
\usetikzlibrary {positioning}
\usepackage{fancyvrb}

\usepackage{balance}

\definecolor {processgray}{cmyk}{0,0,0,0.5}

\newcommand{\items}[0]{\ensuremath{\mathcal{I}}\xspace}
\newcommand{\order}[0]{\ensuremath{\prec}\xspace}
\newcommand{\lang}[0]{\ensuremath{\mathcal{L}}\xspace}
\newcommand{\database}[0]{\ensuremath{\mathcal{D}}\xspace}
\newcommand{\trans}[0]{\ensuremath{t}\xspace}
\newcommand{\litteral}[0]{\ensuremath{X}\xspace}
\newcommand{\card}[1]{\ensuremath{{|{#1}|}}\xspace}
\newcommand{\weight}[0]{\ensuremath{\omega}\xspace}
\newcommand{\utilityTrans}[0]{\ensuremath{u_{\database}}\xspace}
\newcommand{\utility}[0]{\ensuremath{\mathbb{U}_{\lang}}\xspace}
\newcommand{\quantity}[2]{\ensuremath{q(#1,#2)}\xspace}
\newcommand{\price}[1]{\ensuremath{p(#1)}\xspace}
\newcommand{\prob}[0]{\ensuremath{\mathbb{P}}\xspace}
\newcommand{\Cnk}[2]{\ensuremath{\binom{#1}{#2}}\xspace}
\newcommand{\maxLen}[0]{\ensuremath{M}\xspace} 
\newcommand{\minLen}[0]{\ensuremath{\mu}\xspace} 
\newcommand{\utilityLength}[2]{\ensuremath{\mathbb{L}_{#1}(#2)}\xspace}
\newcommand{\uLen}[0]{\ensuremath{\ell}\xspace}
\newcommand{\measure}[0]{\ensuremath{m}\xspace}
\newcommand{\pattern}[0]{\ensuremath{\varphi}\xspace} 
\newcommand{\complexity}[0]{\ensuremath{O}\xspace}  
\newcommand{\vutu}[0]{\ensuremath{\mathcal{V}}\xspace}  
\newcommand{\totalwu}[0]{\ensuremath{W}\xspace}

\newcommand{\knowledgeGraph}[0]{\ensuremath{\mathcal{K}}\xspace}  
\newcommand{\TBox}[0]{\ensuremath{\mathcal{T}}\xspace}  
\newcommand{\ABox}[0]{\ensuremath{\mathcal{A}}\xspace}
\newcommand{\predicate}[0]{\ensuremath{P}\xspace}

\newcommand{\subjectSet}[0]{\ensuremath{\mathcal{S}}\xspace}  
\newcommand{\objectSet}[0]{\ensuremath{\mathcal{O}}\xspace}
\newcommand{\predicateWeight}[0]{\ensuremath{w}\xspace}

\newcommand{\floor}[1]{\lfloor #1 \rfloor}

\newcommand{\myalgo}[0]{\ensuremath{{\textrm{\sc QPlus}}}\xspace}
\newcommand{\myalgoOff}[0]{\ensuremath{{\textrm{\sc QPlusDisk}}}\xspace}
\newcommand{\HAISampler}[0]{\ensuremath{{\textrm{\sc HAISampler}}}\xspace}
\newcommand{\HUPSampler}[0]{\ensuremath{{\textrm{\sc HUPSampler}}}\xspace}
\newcommand{\KGSUMVIZ}[0]{\ensuremath{{\textrm{\sc KGsumviz}}}\xspace}

\newcommand{\Bootstrap}[0]{\ensuremath{{\textrm{\sc Bootstrap}}}\xspace}

\newcommand{\DOREMUS}[0]{\texttt{DOREMUS}\xspace}
\newcommand{\BENICULTURALI}[0]{\texttt{BENICULTURALI}\xspace}
\newcommand{\DBpedia}[0]{\texttt{DBpedia}\xspace}

\newcommand{\retail}[0]{\texttt{Retail}\xspace}
\newcommand{\bmsDeux}[0]{\texttt{BMS2}\xspace}
\newcommand{\chainstoreFIM}[0]{\texttt{ChainstoreFIM}\xspace}
\newcommand{\USCensus}[0]{\texttt{USCensus}\xspace}
\newcommand{\kosarak}[0]{\texttt{Kosarak}\xspace}
\newcommand{\tfiftyItenDtwentyM}[0]{\texttt{t17Mi20Md}\xspace}

\newtheorem{definition}{Definition}
\newtheorem{example}{Example}
\newtheorem{property}{Property}
\newtheorem{theorem}{Theorem}

\begin{document}
\title{
    Scalable Sampling for High Utility Patterns 
}%

 \author{
     \IEEEauthorblockN{Lamine Diop\IEEEauthorrefmark{1}, Marc Plantevit\IEEEauthorrefmark{1}}
     \IEEEauthorblockA{
         \IEEEauthorrefmark{1}EPITA Research Laboratory (LRE), Le Kremlin-Bicetre, Paris FR-94276, France, 
         firstname.lastname@epita.fr
     }
}

\maketitle

\begin{abstract}
Discovering valuable insights from data through meaningful associations is a crucial task. However, it becomes challenging when trying to identify  representative  patterns in quantitative databases, especially with large datasets, as enumeration-based strategies struggle due to the vast search space involved. To tackle this challenge, output space sampling methods have emerged as a promising solution thanks to its ability to discover valuable patterns with reduced computational overhead.  However, existing sampling methods often encounter limitations when dealing with large quantitative database, resulting in scalability-related challenges. In this work, we propose a novel high utility pattern sampling algorithm and its on-disk version both designed for large quantitative databases based on two original theorems. Our approach ensures both the interactivity required for user-centered methods and strong statistical guarantees through random sampling. Thanks to our method, users can instantly discover relevant and representative utility pattern, facilitating efficient exploration of the database within seconds. To demonstrate the interest of our approach, we present a compelling use case involving archaeological knowledge graph sub-profiles discovery. Experiments on semantic and none-semantic quantitative databases show that our approach outperforms the state-of-the art methods.
\end{abstract}


\section{Introduction}
\label{sec:Introduction}

Exploratory Data Analysis (EDA) has faced significant challenges due to the increasing prevalence of large, unfamiliar, and complex datasets. Recently, researchers have proposed methods to construct and explore database profiles for knowledge graphs \cite{alva2021abstat,PrincipeMPCS22,DiopJOCCH23} for an interactive mining and visualization. These profiles effectively capture vital information from the knowledge graph by assigning weights to the predicates connecting two classes (or links between nodes) based on the number of instances belonging to these classes. Consequently, the resulting profile offers an easier-to-understand and interpret representation for end-users\footnote{\url{http://abstat.disco.unimib.it/home} and \url{https://github.com/DTTProfiler/DTTProfiler}} and more easier when it is presented as a network. Networks often provide interactive capabilities, allowing analysts to zoom, pan, and filter the visualization, but, visualizing profiles derived from large knowledge graphs remains a challenging task when employing existing online tool. This is because interactivity empowers users to explore patterns at various levels of granularity while taking account the predicate utilities and the user-defined weights as an external utility. In EDA, the discovery of useful weighted patterns in database with weighted items, called quantitative database (qDB), is known as High Utility Pattern (HUP) and high average-utility (HAUP) mining \cite{Ma2021,GargM0GK23}. This approach extends frequent pattern mining \cite{Agrawal:Apriori94, Han2006MiningFP} to qDB, taking into account both the item quality and quantity in a transaction. HUP mining enhances the interpretability of complex data relationships while visualization techniques, like networks, effectively convey the essence of HUPs as shown in Figure \ref{fig:qplusFrame}. However, it requires rapid, iterative information exchange between the system and the user \cite{Leeuwen14}. Traditional techniques often struggle to meet this time constraint due to the exponential growth in data volume and the output size that can become impractical for human users to process effectively \cite{Qiu2014YAFIMAP, Zaiane2001FastPA}. Other methods, such as those based on condensed representation or top-k patterns, aim to identify optimal patterns but frequently encounter challenges in maintaining pattern diversity \cite{Tseng2015, Tseng2016}.

\begin{figure}
	\centering
	\includegraphics[width=0.48\textwidth]{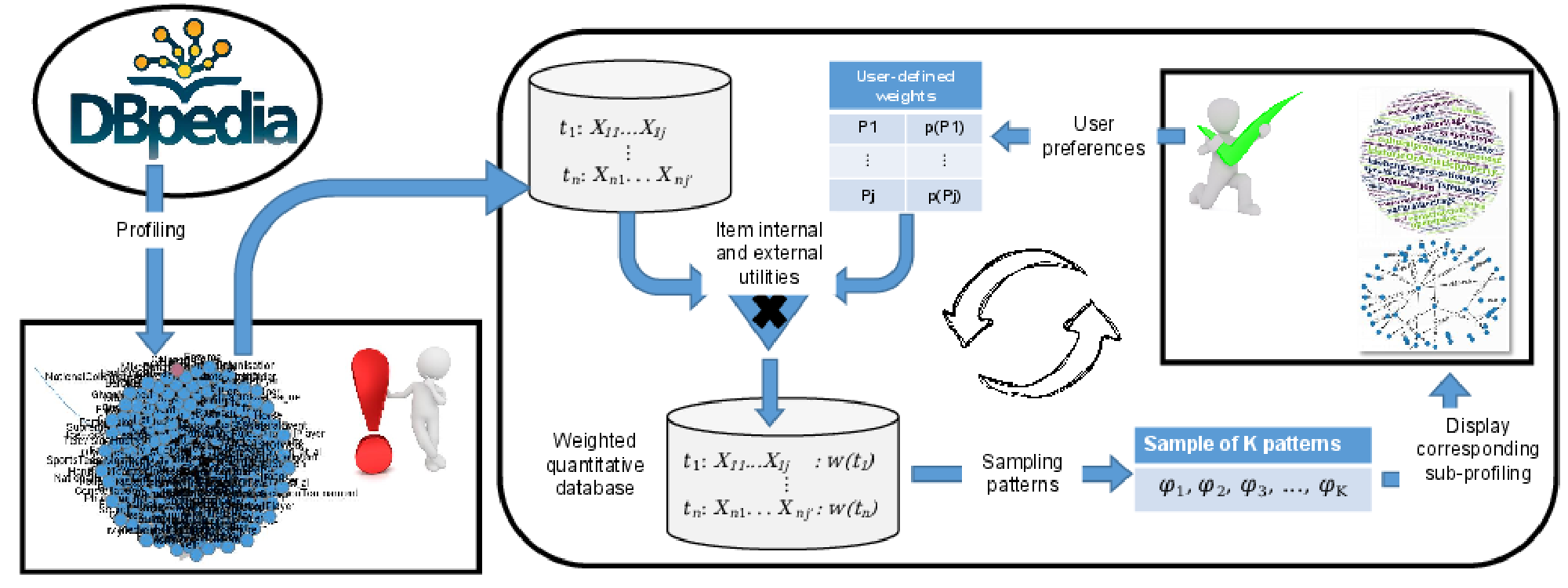}
	\caption{Discovery of sub-profiles in knowledge graphs}
	\label{fig:qplusFrame}
\end{figure}

To overcome these obstacles, researchers have proposed new approaches such as sampling-based approximate mining methods \cite{RiondatoKDD18,PretiKDD21} and output sampling methods \cite{al2009output,BLPG11,diopICDM18,DiopIEEEBigData22}. These innovative techniques enable the extraction of meaningful patterns based on a probability distribution, ensuring both representativeness and control over the output space. Despite their widespread application in frequent pattern mining tasks, they have seen limited utilization within the domain of qDB \cite{DiopPAKDD2022,HUPSampler2023}. This discrepancy can be attributed to the intricate and memory-intensive nature of the weighting phase involved in handling transactions within qDB. Consequently, the complexity of these methods hinders their scalability for large databases. This limitation becomes particularly problematic when exploring very large databases and motivate us to solve the following problems: \textit{How can the response time of a system be improved for knowledge exploration when the system is facing memory constraints?}

The main contributions of our work are as follow:

\begin{itemize}
    \item \myalgo, an efficient algorithm, is proposed to sample HUP/HAUP with or without an interval constraint on the length. 
    We show that \myalgo performs an exact pattern sampling, and we analyze its complexity. For qDBs that cannot fit in memory, \myalgoOff, an on-disk version, is proposed. Experiments on several large real qDBs show that our approach is efficient enough to return hundreds of HUP/HAUP per second where current approaches fail.

    \item We propose two new and original theorems to efficiently compute the total sum of the utility patterns of a given quantitative transaction. In a practical sense, our offering is the concept of Upper Triangle Utility (UTU), a non-materialized upper triangle matrix (storage complexity ``$\complexity(UTU)=0$'') that can be used to deduce the weights of any given quantitative transaction. Based on the theorems and properties stemming from the UTU, we solve the problem of exploring large qDB.

    \item Finally, we show their application in weighted knowledge graph profiles to uncover representative sub-profiles. Let us notice right now that constructing sub-profiles through sampling high utility patterns requires transforming the corresponding knowledge graph profile into a qDB.

\end{itemize}

The outline of this paper is as follows. Section~\ref{sec:Related} reviews related works about HUP/HAUP mining approaches and pattern sampling methods. Section~\ref{sec:Problem} introduces basic definitions and the formal problem statement. We present in Section~\ref{sec:Contributions} our contribution for pattern sampling in qDB. We evaluate our approach in Section~\ref{sec:Experiments} and conclude in Section~\ref{sec:Conclusion}.

\section{Related Work}
\label{sec:Related}
This section presents the HUP/HAUP mining in qDBs literature and output space pattern sampling.

\subsection{HUP and HAUP mining}
Mining HUP/HAUP from qDB poses challenges, including restricting candidate patterns and the computational cost of utility calculation. Various efficient methods have been proposed to address these challenges \cite{Yao04afoundational,Shie2011,Singh2019}. However, exhaustive mining scalability is hindered for large-scale databases with diverse information. The increasing complexity with information diversity makes exploring such databases challenging. Exact HUP mining faces the long tail problem, where low-weighted items in lengthy patterns unexpectedly exhibit high utilities. Approaches like average utility measures \cite{LIN2016} aim to tackle this, but may inadvertently favor specific pattern lengths, leading to biased results. Despite challenges, HUP mining is crucial for extracting valuable insights from qDBs. To enhance efficiency, novel algorithms and sampling techniques \cite{DiopPAKDD2022,HUPSampler2023} have been proposed, offering unique strengths. By incorporating these advancements, researchers can better handle qDB complexities, improving data mining and decision-making processes.

\subsection{Output space HUP/HAUP sampling}
Recently, the \HAISampler algorithm \cite{DiopPAKDD2022} extends output pattern sampling to qDB. Despite its success, scalability in large databases is hindered by intensive weight calculations, making it computationally demanding and memory-intensive. More recently, \HUPSampler \cite{HUPSampler2023}, a hybrid approach, extracts High Utility Patterns (HUP) from qDBs, introducing interval constraints and a random tree-based pattern growth technique. While similar to \HAISampler\cite{DiopPAKDD2022}, \HUPSampler uniform pattern drawing during sampling can compromise representativeness while the phase of building a tree remains time consuming with large databases.

In this paper, we propose a novel pattern sampling method addressing scalability challenges in large qDBs. Our approach maintains efficient and exact sampling with the flexibility to work with or without length constraints. By drawing each pattern with a probability proportional to its utility, we achieve accurate and representative pattern extraction without the need for intensive item weighting. This makes our approach more suitable for large-scale qDBs, preserving performance and efficiency even they exceed main memory capacity.
\section{Problem Statement}
\label{sec:Problem}

This section introduces fundamental concepts and notations, providing the necessary definitions to help readers.

\subsection{Basic definitions and notations}
A knowledge graph is a semantic graph combining a TBox (Terminological Box) and an ABox (Assertional Box). The TBox defines the schema, including classes, properties, and constraints. The ABox contains the actual data, comprising instances of classes and the relationships between them. 
\paragraph{From knowledge graph profile to qDB} Given a knowledge graph \knowledgeGraph built upon a TBox \TBox and an ABox \ABox, i.e., $\knowledgeGraph = (\TBox, \ABox)$, a profile is defined as a set of pairs $((\subjectSet, \predicate, \objectSet), \predicateWeight)$ such that:
\subjectSet is a set of subjects that share a same number of \predicateWeight instances of the set of objects \objectSet via the predicate \predicate. In other words, a profile encodes the types (elements of \TBox) and the terms (instances of thesauri) of a knowledge graph into a compact structure while aggregating their statistics like the number of triples. 

\begin{example}
Let $P_1$=``r:represent'', $P_2$=``r:encompass'', $P_3$=``r:involve'' be predicates, $C_1$=``r:Encounter\_Event'', $C_2$=``r:Man-Made'', $C_3$=``r:Document'', $C_4$=``r:Site'' concepts, and $e_1$=``anastylosis'', $e_2$=``geoarchaeology'', $e_3$=``excavation'' terms. We build in Figure \ref{fig:toy_graphMaxNodes} a toy profile.

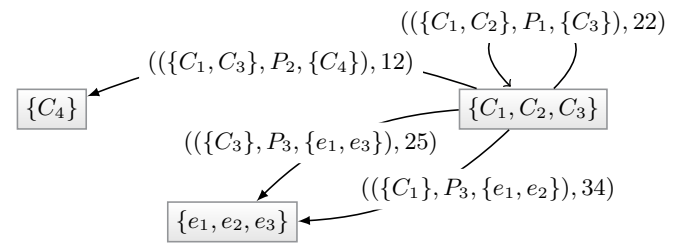
\begin{figure}[htp!]
    {\small
	\begin{center}
		\begin{tikzpicture}[-latex ,auto ,node distance =0 cm and 6.4cm ,on grid ,
		semithick ,
		state/.style ={ top color =white , bottom color = processgray!20 ,
			draw,processgray , text=black , minimum width =0.5cm}, transform shape,mylabel/.style={thin, draw=white, align=center, minimum width=0.5cm, minimum height=0.5cm,fill=white}]
		\node[state] (n1) {$\{C_1,C_2,C_3\}$};
		\node[state] (n2) [above left=of n1] {$\{C_4\}$}; 
		\node[state] (n3) [below left=1.5cm and 4cm of n1] {$\{e_1, e_2, e_3\}$};
		
		\path (n1) edge [loop] node[above =-0.3 cm, mylabel] {$((\{C_1, C_2\},P_1,\{C_3\}), 22)$} (n1);
		\path (n1) edge [bend left =20] node[below left=-.1 cm and -1cm of n1, mylabel] {$((\{C_1\},P_3,\{e_1,e_2\}), 34)$} (n3);	
		\path (n1) edge [bend left =-20] node[below left=.1cm and .5cm of n1, mylabel] {$((\{C_3\},P_3,\{e_1,e_3\}), 25)$} (n3);
		\path (n1) edge [bend left =-20] node[above =-.4 cm, mylabel] {$((\{C_1, C_3\},P_2,\{C_4\}), 12)$} (n2);	
		
		\end{tikzpicture}
	\end{center}
	\caption{Toy knowledge graph profile}
	\label{fig:toy_graphMaxNodes}
    }
\end{figure}

\end{example}

 MCMC-based methods \cite{al2009output}  have been used to address output sampling in unweighted graph data, but they often suffer from time-consuming convergence. An alternative, Random Sampling, lacks guarantees on representativeness, as our baseline \Bootstrap. To overcome these problems, our transformation approach effectively captures the semantic relations within the entire knowledge graph.

To apply our approach, we need first to convert the knowledge graph profile into a qDB \database. Therefore, we consider the weighted items of each quantitative transaction as the set of pairs $((\subjectSet, \predicate, \objectSet), \predicateWeight)$ having a same node source (or target) in \knowledgeGraph. 
Let $\mathcal{I}=\{\litteral_1, \cdots, \litteral_N\}$ be a finite set of items with an arbitrary total order $\order$ between them: $\litteral_1 \order \cdots \order \litteral_N$. A pattern, denoted as $\pattern=\litteral_{j_1} \cdots \litteral_{j_n}$, where $n\leq N$, refers to a non-empty subset of $\mathcal{I}$, i.e., $\pattern \subseteq \mathcal{I}$. The pattern language corresponds to $\mathcal{L} = 2^{{\mathcal{I}}}\setminus\emptyset$, and the length of a pattern $\pattern=\litteral_{j_1} \cdots \litteral_{j_n}$, $\pattern \in \mathcal{L}$, denoted as $\card{\pattern}=n$.

\begin{definition}[Semantic quantitative transaction and qDB]
    A semantic quantitative transaction is a set of weighted items $\trans=\{\litteral_{j_1}\!\!:\!\!\weight_{j_1}  \cdots \litteral_{j_n}\!\!:\!\!\weight_{j_n}\}$ sharing either the same node source or the same node target in the knowledge graph profile, with $\weight_{j_k}$ the quantity of $\litteral_{j_k}$ in $\trans$ denoted by $\quantity{\litteral}{\trans}$. 
A qDB $\mathcal{D}$ corresponds to a multi-set of quantitative transactions. 
\end{definition}

We denote $\trans^i=\litteral_{j_1} \cdots \litteral_{j_i}$ as an itemset formed by the $i$ first items of $\trans$. Thus, we have $\card{\trans^i}=i$. The $i^{th}$ item of the transaction $\trans$ is denoted by $\trans[i]$.

 \begin{example} 
 
 Here are two quantitative transactions that can be generated from the toy profile in Figure \ref{fig:toy_graphMaxNodes}: $\trans_1=\{\litteral_1\!\!:\!\!22, \litteral_2\!\!:\!\!12, \litteral_3\!\!:\!\!25, \litteral_4\!\!:\!\!34\}$ with outgoing predicates, and $\trans_2=\{\litteral_1\!:\!22\}$ with incoming predicates, with $\litteral_1=(\{C_1, C_2\},P_1,\{C_3\})$, $\litteral_2=(\{C_1, C_3\},P_2,\{C_4\})$, $\litteral_3=(\{C_3\},P_3,\{e_1,e_3\})$, $\litteral_4=(\{C_1\},P_3,\{e_1,e_2\})$. The remaining transactions are $\trans_3=\{\litteral_2\!\!:\!\!12\}$ and $\trans_4=\{\litteral_3\!\!:\!\!25, \litteral_4\!\!:\!\!34\}$, then our toy qDB is $\database=\{\trans_1, \trans_2, \trans_3, \trans_4\}$.

\end{example}


\paragraph{User-defined Weights}
In refining subprofile discovery through pattern sampling, we incorporate user-defined strict positive weights for each predicate $P$ of a given item $\litteral$ denoted by $\price{\litteral}$, allowing users to tailor and bias the subprofile view based on individual priorities. This addition provides flexibility to our profile visualization tool, enabling users to customize displayed patterns for a personalized and adaptable exploration experience. Let us consider, in our example, that $\price{\litteral_1}=2$, $\price{\litteral_2}=1$, and $\price{\litteral_3}=\price{\litteral_4}=3$.

\begin{definition}[Pattern utility in a transaction]
\label{def:patternWeightTrans}
Given a transaction \trans of a qDB \database, the weight of the item \litteral in \trans, denoted as $\weight(\litteral, \trans)$, is given by: $\weight(\litteral, \trans)= \quantity{\litteral}{\trans} \times \price{\litteral}.$ 
The utility of a pattern \pattern of $\lang(\database)$ in \trans is defined as follows:

$\utilityTrans(\pattern, \trans)  = \sum_{\litteral \in \pattern} \weight(\litteral, \trans) \text{ if } \pattern \subseteq \trans  \text{ and } 0 \text{ otherwise.}$
\end{definition}

\paragraph{From utility patterns to knowledge graph sub-profile} After that, each sampled pattern is converted into a sub-profile graph where the nodes sharing at least one type or term are merged to form a maximal profile as defined in \cite{DiopJOCCH23}. For instance, if the pattern $\pattern=\litteral_2\litteral_4$ is drawn, then it corresponds to the knowledge graph sub-profile in Figure \ref{fig:subprofile}. 

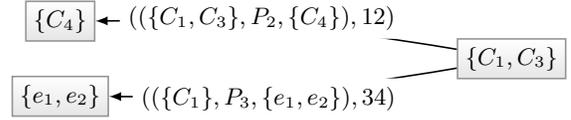
\begin{figure}[htp!]
    {\small
    \begin{center}
    \begin{tikzpicture}[-latex ,auto ,node distance =0 cm and 6.5cm ,on grid ,
    semithick ,
    state/.style ={ top color =white , bottom color = processgray!20 ,
        draw,processgray , text=black , minimum width =0.5cm}, transform shape,mylabel/.style={thin, draw=white, align=center, minimum width=0.5cm, minimum height=0.5cm,fill=white}]
    \node[state] (n1) {$\{C_1,C_3\}$};
    \node[state] (n2) [above left=0.5cm and 6cm of n1] {$\{C_4\}$}; 
    \node[state] (n3) [below left=0.5cm and 6cm of n1] {$\{e_1, e_2\}$};
    
    \path (n1) edge [bend left =5] node[above left=-.1 cm and .2cm of n1, mylabel] {$((\{C_1\},P_3,\{e_1,e_2\}), 34)$} (n3);	
    \path (n1) edge [bend left =-5] node[above left=0.1cm and .2cm of n1, mylabel] {$((\{C_1, C_3\},P_2,\{C_4\}), 12)$} (n2);	
    
	\end{tikzpicture}
	\end{center}
\caption{Sub-profile from the pattern $\litteral_2\litteral_4$}
\label{fig:subprofile}
}
\end{figure}

Based on the previous basic notions, we can formalize the problem of output space pattern sampling in qDBs.

\subsection{HUP/HAUP sampling problem}
    By definition, a pattern sampling method aims to randomly select a pattern \pattern from a language \lang based on an interestingness measure \measure. The notation $\pattern \sim \prob(\lang)$ represents the selection such a pattern, where $\prob(\cdot) = \measure(\cdot) / Z$ is a probability distribution over \lang, and $Z$ is the normalization constant.
    In our case, we specifically focus on high utility as an intuitive measure of interestingness, which allows experts to capture the most representative subsets in the qDB.

\begin{definition}[Utility of a pattern]
\label{def:patternutility}
    Given a qDB \database and a pattern \pattern defined in $\lang(\database)$, the utility of \pattern in \database, denoted as $\utility(\pattern, \database)$, is defined as follows: $\utility(\pattern, \database) = \sum_{\trans \in \database} \utilityTrans(\pattern, \trans).$
\end{definition}

\noindent For generality, we additionally consider utilities that are independent of any specific database, referred to as the length-based utility \cite{DiopPAKDD2022}. We use $\utilityLength{[\minLen..\maxLen]}{\cdot}: \mathbb{N}^* \to \mathbb{R}^+$, where \minLen and \maxLen are two positive integers with $\minLen \leq \maxLen$, to define the length-based utility function. In any case, when this function is used, $\utilityLength{[\minLen..\maxLen]}{\uLen} = 0$ if $\uLen \not\in [\minLen..\maxLen]$. Consequently, if $\utilityLength{[\minLen..\maxLen]}{\uLen} = \frac{1}{\uLen}$, then we are dealing with the sub-problem of HAUP under length constraint \cite{DiopPAKDD2022}, where the utility of a pattern \pattern is defined as its average utility $\frac{\utility(\pattern, \database)}{\card{\pattern} }$ if $\card{\pattern} \in [\minLen..\maxLen]$, and $0$ otherwise. 
With HUP sampling, we set the length-based utility function to $1$ for any pattern in $\lang(\database)$. The maximum length constraint can be set to infinity ($+\infty$) to avoid length constraints. Let $\lang_{[\minLen..\maxLen]}(\database)$ denote the set of all patterns in $\lang(\database)$ with a length greater than \minLen and lower than \maxLen, i.e., $\lang_{[\minLen..\maxLen]}(\database)=\{\pattern \in \lang(\database) : \card{\pattern} \in [\minLen .. \maxLen]\}$.

\begin{example}

Let us compute the utility of $\pattern=\litteral_2\litteral_4$ in \database $\utility(\pattern, \database)= \sum _{\trans \in \database} \utilityTrans(\pattern, \trans)$ with $\utilityTrans(\pattern, \trans) = \sum _{\litteral \in \pattern} \weight(\litteral, \trans)$ if $\pattern \subseteq \trans$ and $0$ otherwise. We have $\weight(\litteral, \trans)= \quantity{\litteral}{\trans} \times \price{\litteral}$ with $\quantity{\litteral}{\trans}$ the quantity of \litteral in transaction \trans and $\price{\litteral}$ the user-defined weight of \litteral. In our case, we have $\price{\litteral_2}=1$ and $\price{\litteral_4}=3$. With $\trans_1=\{\litteral_1\!\!:\!\!22, \litteral_2\!\!:\!\!12, \litteral_3\!\!:\!\!25, \litteral_4\!\!:\!\!34\}$, we get $\utilityTrans(\pattern, \trans_1)= (22\times 1) + (34 \times 3) = 124$ while $\utilityTrans(\pattern, \trans_2)= 0$, $\utilityTrans(\pattern, \trans_3)= 0$, and $\utilityTrans(\pattern, \trans_4)= 0$. Therefore, $\utility(\pattern, \database)=124+0+0+0=124$. The average utility of $\pattern$ in \database is $\frac{\utility(\pattern, \database)}{\card{\pattern}} = \frac{124}{2}$. If we consider an interval constraint $[\minLen..\maxLen]$, $\minLen=1$ and $\maxLen=2$, with a length-based utility $\utilityLength{[\minLen..\maxLen]}{\uLen}=\frac{1}{\uLen}$, $\utility(\pattern, \database) \times \utilityLength{[\minLen..\maxLen]}{\card{\pattern}}$ is the average utility of $\pattern$, while $\utilityLength{[\minLen..\maxLen]}{\card{\litteral_2\litteral_3\litteral_4}} =0$ because $\card{\litteral_2\litteral_3\litteral_4}=3 \not\in [\minLen..\maxLen]$. 
\end{example}




Finally, {given a qDB $\database$, two positive integers $\minLen$ and $\maxLen$, our objective is to draw a pattern $\pattern \in \lang_{[\minLen .. \maxLen]}$ with a probability proportional to its weighted utility in the  database $\database$, i.e.,
	$$\prob(\pattern) = \frac{\utility(\pattern, \database) \times \utilityLength{[\minLen .. \maxLen]}{\card{\pattern}}}{\sum_{\pattern' \in \lang_{[\minLen .. \maxLen]}} \utility(\pattern', \database) \times \utilityLength{[\minLen .. \maxLen]}{\card{\pattern'}}}.$$
}


Table \ref{tab:NotationsSummary} summarizes the notations used in this paper.

\begin{table}
	\caption{Notations}	
	\label{tab:NotationsSummary}
	{\small
	\begin{tabular}{l|l}
		\hline
		Symbol & Denotes \\
		\hline
		qDB & Quantitative database\\
		\items & Set of all items $\{\litteral_1, \cdots, \litteral_n\}$ \\
		\lang &  Pattern language $2^{\items}\setminus\emptyset$ \\
		\pattern &  Pattern, $\pattern \in \lang$ \\
		\trans &  Transaction $\trans \subseteq \items$\\
		\database &  Database, a multi-set of transactions \\
		$\trans[i]$ &  The $\text{i}^{\text{th}}$ item of transaction \trans \\
		$\trans^{i}$ & An itemset composed by the $i$ first items of \trans  \\
		$\quantity{\litteral}{\trans}$ & The quantity of item \litteral in transaction \trans \\
		$\price{\litteral}$ & The user-defined weight on item \litteral \\
		\TBox & A TBox called ontology, set of triples \\
		\ABox & An Abox, the instances, set of triples \\
		$\knowledgeGraph$ & A knowledge graph defined by $\knowledgeGraph = (\TBox, \ABox)$\\
		\subjectSet & Subject, a set of concepts (element of \TBox) \\
		\objectSet & Object, a set of concepts or terms \\
		\predicate & A predicate, used to link two instances of \ABox \\
		$\lang(\database)$ &  The set of all patterns of \lang defined in \database \\
		$\weight(\litteral, \trans)$ & The weight of \litteral in \trans defined by $\quantity{\litteral}{\trans} \times \price{\litteral}$ \\
		$\utilityTrans(\pattern, \trans)$ & Utility of pattern \pattern in transaction \trans \\
		$\utility(\pattern, \database)$ & Utility of pattern \pattern in the entire database \database \\
		$\utilityLength{[\minLen..\maxLen]}{\cdot}$ & Length-based utility measure \\
		\minLen & Minimum length constraint \\
		\maxLen & Maximum length constraint \\
		$\vutu_\trans$ & The upper triangle utility of transaction \trans \\
        $\vutu_\trans(\uLen, i)$ & The sum of utilities of all patterns of length \uLen in $\trans^i$ \\
		\hline
	\end{tabular}
	}
\end{table}

\section{Generic utility pattern sampling}
\label{sec:Contributions}
We adopt in this paper the multi-step pattern sampling \cite{BLPG11} technique. Therefore, we first need to perform a preprocessing phase, which involves weighting each transaction.
\subsection{Weighting of a qDB}
\label{sub:counting}

Before introducing our weighting approach, let us first briefly explain our baseline weighting approach.

\paragraph{\HAISampler \cite{DiopPAKDD2022} weighting approach} The transaction weighting phase is one of the most fundamental steps for multi-step pattern sampling. It involves computing the sum of the utility of all the patterns in a given transaction. \HAISampler uses a weighting matrix to compute the weight of a length \uLen, which represents the aggregated weights of the set of patterns of length \uLen that appear in the transaction. Let us now explain how this weighting approach works according to the arbitrary total order relation \order.
For a given transaction \trans, each item $\trans[i]$ is associated with a local matrix of 2 rows and $L$ columns, where $L$ is the maximal length of the pattern that can be generated with the first $i^{th}$ items ($\trans^i$) according to \order and the maximal length constraint. The first row of this local matrix contains the weight of the patterns with item $\trans[i]$ in transaction $\trans^i$, while the second row contains those patterns that do not contain $\trans[i]$.
The weighting matrix of transaction $\trans_1$ is computed in Table \ref{tab:weightinHAISampler}. For the sake of simplicity, we omit columns containing zeros only. The local matrix at item $\litteral_3=\trans_1[3]$ shows that, in the transaction $\trans_1^3 = \{\litteral_1\!\!:\!\!22, ~\litteral_2\!\!:\!\!12, ~\litteral_3\!\!:\!\!25\}$, the sum of the utilities of the patterns of length 2 with (without) item $\trans_1[3]=\litteral_3$ is $206$ (respectively $56$). The sum of utilities of patterns of length 3 with $\trans_1[3]$ is 131, while there are no patterns of length 3 without $\trans_1[3]$. Finally, the weight of this transaction is calculated as $(34+131)+(233+262)+(364+131)+(165+0)=1320$.

\begin{table}[htp!]
    \centering 
    \caption{Weighting matrix of $\trans_1$ by \HAISampler \cite{DiopPAKDD2022}}
    \label{tab:weightinHAISampler}
    \resizebox{\linewidth}{!}{%
        \begin{tabular}{c c c cc c ccc c cccc}
            $\trans_1:\{$ & \multicolumn{1}{c}{$\litteral_1$}& , & \multicolumn{2}{c}{$\litteral_2$} & , & \multicolumn{3}{c}{$\litteral_3$} & , & \multicolumn{4}{c}{$\litteral_4$} $\quad\}$\\
            
            & \multicolumn{1}{c}{44}
            &  & \multicolumn{1}{c|}{12} & \multicolumn{1}{r}{56}
            &  & \multicolumn{1}{c|}{75} & \multicolumn{1}{r|}{206} & \multicolumn{1}{c}{131}
             &  & \multicolumn{1}{c|}{34} & \multicolumn{1}{r|}{233} & \multicolumn{1}{c|}{364} & \multicolumn{1}{c}{165}\\
            \cline{2-2} \cline{4-5} \cline{7-9} \cline{11-14}
            & \multicolumn{1}{c}{0}
            &  & \multicolumn{1}{c|}{44} & \multicolumn{1}{r}{0}
            &  & \multicolumn{1}{c|}{56} & \multicolumn{1}{r|}{56} & 0
             &  & \multicolumn{1}{c|}{131} & \multicolumn{1}{r|}{262} & \multicolumn{1}{c|}{131} & \multicolumn{1}{c}{0} \\
        \end{tabular}
    }
\end{table}

After weighting all transactions in the database \database, \HAISampler can efficiently draw a sample of patterns where each one is drawn proportionally to its utility in \database. However, in large databases with limited available memory, the storage cost can become prohibitive. Therefore, we propose a new, efficient, and original weighting approach for qDB.

\paragraph{Our weighting approach} The solution we propose in this paper eliminates the need to compute a local matrix for each item. Our main goal for the weighting is to establish a formula that can compute the weight of any transaction in a qDB based on the utility of the transaction, which is the sum of the utility of its items. To achieve this, we introduce the Upper Triangle Utility weighting approach. It utilizes an upper triangle matrix containing the weights of each item, enabling the computation of the transaction weight and an exact draw of a pattern. Specifically, at the intersection of line \uLen and column $i$, we have the sum of the utilities of the set of patterns of length \uLen that appear in the transaction portion $\trans^i$. Here is the formal definition of Upper Triangle Utility.

\begin{definition}[Upper Triangle Utility]
\label{def:VUTU}
Given a quantitative transaction \trans, \uLen and $i$ two integers, with $0<\uLen, i \leq \card{\trans}$. The upper triangle utility of \trans denoted by $\vutu_\trans$ is defined as follows:

$\vutu_\trans(\uLen, i) = 0 ~if ~\uLen > i \text{, \qquad} \vutu_\trans(1, i) = \sum_{j=1}^{i}\weight(\trans[j], \trans)  $ 
$\vutu_\trans(\uLen, i) = \Cnk{i-1}{\uLen-1} \times\weight(\trans[i], \trans) + \vutu_\trans(\uLen-1, i-1) + \vutu_\trans(\uLen, i-1).$

\end{definition}

\noindent The recursive formula can be interpreted easily as follows: on one hand, $\Cnk{i-1}{\uLen-1} \times \weight(\trans[i], \trans) + \vutu_\trans(\uLen-1, i-1)$ gives the aggregate weights of the set of patterns of length $\uLen$ appearing in $\trans$ and containing the item $\trans[i]$. On the other hand, $\vutu_\trans(\uLen, i-1)$ gives that of the set of patterns of length $\uLen$ in $\trans$ without the item $\trans[i]$.

\begin{example}
    Let us compute the upper triangle utility of $\trans_1$ in Table \ref{tab:t5_VUTU} with Definition \ref{def:VUTU}.
    \begin{table}[htp!]
    \centering 
    \caption{Upper triangle utility of transaction $\trans_1$}
    {\small
    \label{tab:t5_VUTU}
		\begin{tabular}{c c c c c c c c c c}
			$\trans_1:\{$ & $\litteral_1$ & , & $\litteral_2$  & , & $\litteral_3$  & , & $\litteral_4$  & $\}$ \\
			
			& \cellcolor{gray!20} 44 &\cellcolor{gray!20} & \cellcolor{gray!20} 56 & \cellcolor{gray!20} & \cellcolor{gray!20} 131 & \cellcolor{gray!20} & \cellcolor{gray!20} 165 & \\
			& 0 & & \cellcolor{gray!20} 56 & \cellcolor{gray!20} & \cellcolor{gray!20} 262 & \cellcolor{gray!20} & \cellcolor{gray!20} 495 & \\
			& 0 & & 0 &  & \cellcolor{gray!20} 131 & \cellcolor{gray!20} & \cellcolor{gray!20} 495 & \\
			& 0 & & 0 & & 0 &  & \cellcolor{gray!20} 165 & \\
		\end{tabular}
    }
    \end{table}
The weight of the corresponding transaction \trans is the sum of the values in column $\card{\trans}$ of $\vutu_\trans$. Then we have $165 + 495 + 495 + 165 = 1320$. Hence we get the same result as \HAISampler.

\end{example}

In reality, the upper triangle utility hides wonderful properties. For example, we have the intuition that $\vutu_\trans(\uLen, i) = \vutu_\trans(i-\uLen+1, i)$. Furthermore, for any couple of values $(\uLen, i)$ such that $\uLen \leq i \leq \card{\trans}$, we get the result $\frac{\vutu_\trans(\uLen, i)}{\vutu_\trans(1, i)} = \binom{i-1}{\uLen-1}$. Hence, we state Theorem \ref{theo:vutu}.

\begin{theorem}
    \label{theo:vutu}
    Let $\vutu_\trans$ be the upper triangle utility of any given a quantitative transaction \trans. For any positive integers \uLen and $i$ such that $\uLen \leq i \leq \card{\trans}$, the following statement holds:
    $\vutu_\trans(\uLen, i) = \Cnk{i-1}{\uLen-1} \times \vutu_\trans(1, i).$
\end{theorem}

\begin{sloppypar}
\begin{proof}(Theorem \ref{theo:vutu})
    Let us prove the formula $\vutu_\trans(\uLen, i) = \Cnk{i-1}{\uLen-1} \times \vutu_\trans(1, i)$ by induction, with \uLen and $i$ two positive integers such that $\uLen \leq i \leq \card{\trans}$. We need to show that it holds for the base case $\uLen=1$ and then prove the induction step for $\uLen>1$.

    \noindent {\bf Base case (\uLen = 1):} For $\uLen=1$, we have $\vutu_\trans(1, i)=\sum_{j=1}^{i}\weight(\trans[j], \trans)$. According to Definition \ref{def:VUTU}, we can rewrite this as: $\vutu_\trans(1, i) = \Cnk{i-1}{1-1} \times \vutu_\trans(1, i)$. Since $\Cnk{i-1}{0}=1$, we can see that the base case holds.

    \noindent {\bf Inductive Step:} Assume that the formula holds for some $\uLen=k$, i.e., $\vutu_\trans(k, i) = \Cnk{i-1}{k-1} \times \vutu_\trans(1, i)$. We will prove that it holds for $\uLen=k+1$, i.e., $\vutu_\trans(k+1, i) = \Cnk{i-1}{k} \times \vutu_\trans(1, i)$.
    Using the recursive equation we provided in Definition \ref{def:VUTU}, we have: $\vutu_\trans(k+1, i) = \Cnk{i-1}{k} \times\weight(\trans[i], \trans) + \vutu_\trans(k, i-1) + \vutu_\trans(k+1, i-1)$. Substituting the induction hypothesis $\vutu_\trans(k, i) = \Cnk{i-1}{k-1} \times \vutu_\trans(1, i)$ and simplifying, we get: $\vutu_\trans(k+1, i) = \Cnk{i-1}{k} \times\weight(\trans[i], \trans) + \Cnk{i-2}{k-1}\times \vutu_\trans(1, i-1) + \Cnk{i-2}{k}\times \vutu_\trans(1, i-1)$. This leads to $\vutu_\trans(k+1, i) = \Cnk{i-1}{k} \times\weight(\trans[i], \trans) + (\Cnk{i-2}{k-1} + \Cnk{i-2}{k})\times \vutu_\trans(1, i-1)$. We know that $\Cnk{i-2}{k-1} + \Cnk{i-2}{k} = \Cnk{i-1}{k}$ based on the binomial coefficients formula. Applying this identity, we have: $\vutu_\trans(k+1, i) = \Cnk{i-1}{k} \times\weight(\trans[i], \trans) + \Cnk{i-1}{k}\times \vutu_\trans(1, i-1)$. Therefore, we get: $\vutu_\trans(k+1, i) = \Cnk{i-1}{k} \times (\weight(\trans[i], \trans) + \vutu_\trans(1, i-1))$. We know by definition that $\weight(\trans[i], \trans) + \vutu_\trans(1, i-1) = \vutu_\trans(1, i)$. Then, $\vutu_\trans(k+1, i) = \Cnk{i-1}{k} \times \vutu_\trans(1, i)$. Thus, the equation holds for $\uLen=k+1$. By proving the base case and the induction step, we have shown that the formula $\vutu_\trans(\uLen, i) = \Cnk{i-1}{\uLen-1} \times \vutu_\trans(1, i)$ holds for all \uLen by induction.
\end{proof}
\end{sloppypar}

In addition, we can also state the following property.


\begin{property}[Symmetry]
\label{prop:VUTU_symmetry}
    If $\vutu_\trans$ is the upper triangle utility of a given transaction \trans, and \uLen and $i$ two positive integers such that $\uLen \leq i \leq \card{\trans}$ then we have:
    $\vutu_\trans(\uLen, i) = \vutu_\trans(i-\uLen+1, i).$
\end{property}

\begin{proof}(Property \ref{prop:VUTU_symmetry})
    We know that $\vutu_\trans(\uLen, i) = \binom{i-1}{\uLen-1} \times \vutu_\trans(1, i)$ holds for all $\uLen$ according to Theorem \ref{theo:vutu}. We can derive $\vutu_\trans(i-\uLen+1, i) = \binom{i-1}{i-\uLen+1-1} \times \vutu_\trans(1, i) = \binom{i-1}{i-\uLen} \times \vutu_\trans(1, i)$. Since $\binom{i-1}{\uLen-1} = \binom{i-1}{i-\uLen}$, this property becomes trivial and holds true.
\end{proof}

\noindent Property \ref{prop:VUTU_symmetry} is particularly useful in the preprocessing and the drawing phases, especially when the maximal length constraint is either not used or set higher than half of the transaction length. Therefore, it is used automatically in our approach. In the case where there is no maximal length constraint (i.e., $\maxLen=\infty$) while the minimal one is set to $\minLen=1$, Theorem \ref{theo:totalweights} can be used to compute the weigh of each transaction.


\begin{theorem}
	\label{theo:totalweights}
	Let \trans be a transaction from a qDB. The cumulative utility sum of the entire set of patterns that appear in transaction \trans is expressed as:
	$\totalwu(\trans) = 2^{\card{\trans}-1}\times \sum_{\litteral \in \trans}\weight(\litteral,\trans).$
\end{theorem}

\begin{proof}(Theorem \ref{theo:totalweights})
    Let \minLen and \maxLen be two integers such that $0<\minLen \leq \maxLen$. We know that the aggregate utility of all patterns under interval constraint $[\minLen..\maxLen]$ is $\totalwu_{[\minLen..\maxLen]}(\trans) = \sum_{\uLen=\minLen}^{\maxLen} \vutu_\trans(\uLen, \card{\trans})$. Based on Theorem \ref{theo:vutu}, this formula can be rewritten as follows: $\totalwu_{[\minLen..\maxLen]}(\trans) = \sum_{\uLen=\minLen}^{\maxLen} (\Cnk{\card{\trans}-1}{\uLen-1} \times \vutu_\trans(1, \card{\trans})) = (\sum_{\uLen=\minLen}^{\maxLen} \Cnk{\card{\trans}-1}{\uLen-1}) \times \vutu_\trans(1, \card{\trans})$. If $\maxLen=\infty$ and $\minLen=1$, then we got $\totalwu_{[1..\infty[}(\trans) = \totalwu(\trans) = (\Cnk{\card{\trans}-1}{1-1}+\Cnk{\card{\trans}-1}{2-1}+\cdots+\Cnk{\card{\trans}-1}{\card{\trans}-1})\times \vutu_\trans(1, \card{\trans})$. Since $\Cnk{\card{\trans}-1}{1-1}+\Cnk{\card{\trans}-1}{2-1}+\cdots+\Cnk{\card{\trans}-1}{\card{\trans}-1} =2^{\card{\trans}-1}$ and, by definition, $\vutu_\trans(1, \card{\trans})=\sum_{\litteral \in \trans}\weight(X,\trans)$, then $\totalwu(\trans) = 2^{\card{\trans}-1}\times \sum_{\litteral \in \trans}\weight(X,\trans).$ Hence the result.
\end{proof}



Based on these theorems, we no longer need to store the UTU in memory. This is because we can easily deduce any value $\vutu_\trans(\uLen, i)$ based on $\vutu_\trans(1, i) = \sum_{j=1}^{i}\weight(\trans[j], \trans)$.

After the preprocssing phase, we show in the following section how to sample efficiently a pattern from a transaction proportionally to its weighted utility.
\subsection{Drawing a high utility pattern}
\label{sub:sampling}

Sampling a pattern can be done efficiently from the upper triangle utility. Let us first recall the definition of the probability to select an item at a given position.

\begin{definition}
\label{def:probability_item_selection}
    Let \trans be a quantitative transaction, \pattern be the already drawn $k$ items ordered according to \order, and \uLen be the rest of the items to select, $\pattern=\litteral_{j_{k}}\cdots \litteral_{j_{1}}$. The probability to pick the item $\trans[i]$, with $i<j_{k}$, given \pattern and \uLen, is defined as the ratio of the aggregate weights of all patterns of length $\uLen+k$ in the transaction $\trans^i\cup \pattern$, containing the pattern ${\trans[i]} \cup \pattern$, equals to $\vutu_\trans(\uLen, i) - \vutu_\trans(\uLen, i-1) + \Cnk{i-1}{\uLen-1} \times \utility(\pattern, \trans)$, and the total aggregate weights of all patterns of length $\uLen+k$ appearing in the transaction $\trans^i\cup \pattern$, containing the pattern $\pattern$, equals to $\vutu_\trans(\uLen, i) + \Cnk{i}{\uLen} \times \utility(\pattern, \trans)$. Formally, we have:
    $$\prob(\trans[i]|\pattern \wedge \uLen) = \frac{\vutu_\trans(\uLen, i) - \vutu_\trans(\uLen, i-1) + \Cnk{i-1}{\uLen-1}\times \utility(\pattern, \trans)}{\vutu_\trans(\uLen, i) + \Cnk{i}{\uLen}\times \utility(\pattern, \trans)}.$$
\end{definition}

\noindent This means that a pattern \pattern is drawn proportionally to its utility in a given transaction \trans.

\begin{sloppypar}
\begin{example}
The pattern $\litteral_1\litteral_3$ can be drawn in $\trans_1$ using its upper triangle utility $\vutu_{\trans_1}$ as follows:
$\prob(\litteral_1\litteral_3, \trans_1) = \prob(\uLen=2|\trans_1) \times (1-\prob(\litteral_4|\pattern=\emptyset \wedge \uLen=2)) \times \prob(\litteral_3|\pattern=\emptyset \wedge \uLen=2) \times (1-\prob(\litteral_2|\pattern=\litteral_3 \wedge \uLen=1)) \times \prob(\litteral_1|\pattern=\litteral_2 \wedge \uLen=1)$. Then we have: $\prob(\litteral_1\litteral_3, \trans_1) = \frac{495}{1320} \times (1-\frac{495-262}{495}) \times \frac{262-56}{262} \times (1-\frac{56-44+\Cnk{1}{0} \times \utility(\litteral_3, \trans_1)}{56+\Cnk{2}{1} \times \utility(\litteral_3, \trans_1)}) \times \frac{44+\Cnk{0}{0} \times \utility(\litteral_3, \trans_1)}{44+\Cnk{1}{1} \times \utility(\litteral_3, \trans_1)}$. By simplifying, we get $\prob(\litteral_1\litteral_3, \trans_1) = \frac{44+75}{1320} = \frac{\utility(\litteral_1\litteral_3, \trans_1)}{\sum_{\pattern \subseteq \trans_1} \utility(\pattern, \trans_1)}$.
\end{example}

\end{sloppypar}

However, we can see that this approach needs to generate a random variable at each visited item, and it is possible to visit all items of the transaction during the drawing of a pattern. For instance, if in $\trans_5$ the drawn pattern contains the first item $\trans_5[1]=A$, we need a more efficient approach that does not require generating more than $\card{\pattern}$ random variables. Therefore, we propose an approach based on dichotomous random search.

\begin{property}[Dichotomous random search]
\label{prop:dichotomous_random_search}
Let $\trans$ be a quantitative transaction, and $\trans[j]$ be the last picked item during the drawing process. Using a sequential random variable generation to draw the next item $\trans[i]$, with $i<j$, to add in the sampled pattern $\pattern$ is equivalent to jumping directly on it based on dichotomous search. In other words, if a random number $\alpha_i$ drawn from the interval $]0, \vutu_\trans(\uLen,i)+\Cnk{i}{\uLen}\times \utility(\pattern, \trans)]$ allows us to pick $\trans[i]$, i.e., $\frac{\alpha_i}{\vutu_\trans(\uLen, i) + \Cnk{i}{\uLen}\times \utility(\pattern, \trans)} \leq \prob(\trans[i]|\pattern \wedge \uLen)$, then from position $j$ of $\trans$, we can directly select the item $\trans[i]$ such that $\vutu_\trans(\uLen, i-1)+\Cnk{i-1}{\uLen}\times\utility(\pattern,\trans) < \alpha_i \leq \vutu_\trans(\uLen, i)+\Cnk{i-1}{\uLen-1}\times\utility(\pattern,\trans)$.
\end{property}


\begin{proof}(Property \ref{prop:dichotomous_random_search})
This property is somewhat trivial because, after using Bayes' Theorem and Definition \ref{def:probability_item_selection}, it stems from simplifying the intermediate values which do not allow selecting their corresponding item based on their randomly generated variables.
\end{proof}

The usefulness of Property \ref{prop:dichotomous_random_search} is more valuable when the transactions are very long. The iterative approach used in \HAISampler drawing phase visits in the worst case all the items of the transaction. Such approach is often time consuming with long transactions. Thanks to the dichotomous random search property, we just need to generate the only necessary random variables while jumping directly to the corresponding positions of the transaction.

Now, we can formalize our algorithm \myalgo based on the previous statements for high utility pattern sampling.
\subsection{Overview of the algorithm}
\label{sec:algorithm}

\myalgo takes a qDB $\database$ with a total order relation \order, along with positive integers \minLen and \maxLen. The goal is to output a pattern $\pattern$ drawn proportionally to its utility, which is weighted by a length-based utility measure. It consists of two phases: a preprocessing phase and a drawing phase.

\noindent {\bf In the preprocessing phase:} Compute weights $W_{[\minLen..\maxLen]}(\trans)$ for each transaction $\trans$ in $\database$ based on specific formulas involving utility and length functions.

\noindent {\bf In the drawing phase:} First, it draws a transaction $\trans$ from $\database$ proportionally to its weight, and an integer $\uLen$ from $\minLen$ to $\maxLen$ with a probability proportional to the sum of pattern utilities of length $\uLen$ in $\trans$. Therefore, it iteratively draws items from the transaction $\trans$ based on a random process involving utility and length calculations until a pattern $\pattern$ of length $\uLen$ is formed. Finally, it returns $\pattern$ proportionally to its weighted utility.

\begin{algorithm}
    {\small
        \caption{\myalgo: Quantitative Pattern Sampling}
        \label{alg:Sampling}
        \begin{algorithmic}[1]
            \Statex {\bf Input:} A qDB $\database$ with an arbitrary total order relation \order, and two positive integers \minLen and \maxLen such that $\minLen \leq \maxLen$
            \Statex {\bf Output:} A pattern 
            $\pattern \sim \utility(\lang_{[\minLen..\maxLen]}(\database))$
            
            \Statex //Preprocessing
            \State $W_{[\minLen..\maxLen]}(\trans)  = \sum_{\uLen = \minLen}^{\maxLen} (\vutu_\trans(\uLen, \card{\trans}) \times \utilityLength{[\minLen..\maxLen]}{\uLen})$ for $\trans \in \database$
            \Statex //Drawing
            \State Draw a transaction $\trans$ : $\trans \sim W_{[\minLen..\maxLen]}(\database) = \sum_{\trans \in \database}W_{[\minLen..\maxLen]}(\trans) $
            \State Draw an integer $\uLen$ from $\minLen$ to $\maxLen$ with $\prob(\uLen|\trans)=\frac{\vutu_\trans(\uLen, \card{\trans})\times \utilityLength{[\minLen..\maxLen]}{\uLen}}{W_{[\minLen..\maxLen]}(\trans)}$
            \State $\pattern \gets \emptyset$, $j \gets \card{\trans}$
            \While {$\uLen>0$} \Comment{Process of drawing a pattern of length \uLen}
            \State $\alpha \gets \text{random}(0,1)\times (\vutu_\trans(\uLen,j) + \Cnk{j-1}{\uLen-1}\times \utility(\pattern,\trans))$
            \State $i \gets \underset{i}{\arg}(b_{inf}(i) < \alpha \leq b_{sup}(i))$ 
            \Statex $~~~~ \text{with} ~~b_{inf}(i) = \vutu_\trans(\uLen,i-1) + \Cnk{i-1}{\uLen}\times \utility(\pattern,\trans)$
            \Statex $~~~~\text{and} ~~b_{sup}(i) =\vutu_\trans(\uLen,i) + \Cnk{i-1}{\uLen-1}\times \utility(\pattern,\trans)$
            \State $\pattern \gets \{\trans[i]\} \cup \pattern$ and $\uLen \gets \uLen -1$
            \State $j \gets i$ \Comment{Jumping to the position indexed by i}
            \EndWhile
            \State \Return{\pattern} \Comment{Drawn proportionally to its weighted utility}
        \end{algorithmic}
    }
\end{algorithm}

\subsection{Theoretical analysis of the method}
\label{sub:analysis}

In this section, we first demonstrate the soundness of \myalgo before evaluating its time complexity (preprocessing and the sampling). Finally, we give its storage complexity.

\paragraph{Soundness of \myalgo} Property \ref{prop:soundness} states that our algorithm \myalgo performs an exact draw of a pattern.

\begin{property}[Soundness]
    \label{prop:soundness}
    Let \database be a qDB with an arbitrary total order relation \order, and \minLen and \maxLen be two positive integers such that $\minLen \leq \maxLen$. Algorithm \myalgo directly draws a pattern from \database with a length in the range $[\minLen..\maxLen]$ according to a distribution that is proportional to its weighted utility.
\end{property}

\begin{sloppypar}
\begin{proof}(Property \ref{prop:soundness})
    Given two positive integers \minLen and \maxLen such that $\minLen \leq \maxLen$, let $Z$ be the normalization constant defined by $Z = \sum_{\pattern \in \lang_{[\minLen..\maxLen]}(\database)} \left( \utility(\pattern, \database) \times \utilityLength{[\minLen..\maxLen]}{\uLen} \right)$. Let $\prob_{[\minLen..\maxLen]}(\pattern, \database)=\sum_{\trans \in \database}\prob(\pattern, \trans) = \sum_{\trans \in \database} \prob(\trans) \prob_{[\minLen..\maxLen]}(\pattern | \trans)$ be the probability to draw pattern \pattern from the set of patterns in \database using algorithm \myalgo under the length constraint. Based on line 2, we get $\prob(\trans)=\frac{W_{[\minLen..\maxLen]}(\trans)}{Z}$. From lines 3 and 4, we know that $\prob_{[\minLen..\maxLen]}(\pattern | \trans) = \prob(\uLen | \trans) \times \prob(\pattern | \trans \wedge \uLen)$. Line 3 also shows that $\prob(\uLen | \trans) = \frac{\vutu_\trans(\uLen, \card{\trans})\times \utilityLength{[\minLen..\maxLen]}{\uLen}}{\sum_{\uLen' = \minLen}^{\maxLen} (\vutu_\trans(\uLen', \card{\trans}) \times f_{[\minLen..\maxLen]}(\uLen'))} = \frac{\vutu_\trans(\uLen, \card{\trans})\times \utilityLength{[\minLen..\maxLen]}{\uLen}}{W_{[\minLen..\maxLen]}(\trans)}$.
    We also know that $\prob(\pattern | \trans \wedge \uLen)=\frac{\utilityTrans(\pattern, \trans)}{\sum_{\pattern' \in \trans \wedge \card{\pattern'} \in [\minLen..\maxLen]} \utilityTrans(\pattern', \trans)}=\frac{\utilityTrans(\pattern, \trans) \times \utilityLength{[\minLen..\maxLen]}{\uLen}}{\vutu_\trans(\uLen, \card{\trans}) \times \utilityLength{[\minLen..\maxLen]}{\uLen}}$ from lines 7 and 8. By substitution, we have 
    $\prob_{[\minLen..\maxLen]}(\pattern, \database)=\sum_{\trans \in \database} (\frac{W_{[\minLen..\maxLen]}(\trans)}{Z} \times \frac{\vutu_\trans(\uLen, \card{\trans})\times \utilityLength{[\minLen..\maxLen]}{\uLen}}{W_{[\minLen..\maxLen]}(\trans)} \times \frac{\utilityTrans(\pattern, \trans) \times \utilityLength{[\minLen..\maxLen]}{\uLen}}{\vutu_\trans(\uLen, \card{\trans}) \times \utilityLength{[\minLen..\maxLen]}{\uLen}})= \sum_{\trans \in \database} \frac{\utilityTrans(\pattern, \trans) \times \utilityLength{[\minLen..\maxLen]}{\uLen}}{Z} = \frac{\utilityLength{[\minLen..\maxLen]}{\uLen}}{Z} \times (\sum_{\trans \in \database} \utilityTrans(\pattern, \trans))$. Thanks to Definition \ref{def:patternutility}, $\utility(\pattern, \database) = \sum_{\trans \in \database} \utilityTrans(\pattern, \trans)$. Then, we can conclude that $\prob_{[\minLen..\maxLen]}(\pattern, \database) = \frac{\utilityLength{[\minLen..\maxLen]}{\uLen}}{Z} \times \utility(\pattern, \database)$, which shows that \pattern is drawn proportionally to its weighted utility and completes this proof.
\end{proof}
\end{sloppypar}

We will now analyze the complexity of \myalgo. 

\paragraph{Preprocessing complexity} Suppose that the maximal length $\maxLen$ and the sum of the item utilities are known. If $\maxLen=\infty$, each transaction is weighted in $\complexity(1)$. In that case, \database is preprocessed in $\complexity(\card{\database})$. Otherwise, \myalgo weights each transaction of \database in $\complexity(\maxLen)$. Hence, the complexity of preprocessing \database is $\complexity(\card{\database} \times \maxLen)$. In comparison, \HAISampler weights \database in $\complexity(2 \times \card{\database}\times\card{\items}\times (\maxLen-\minLen))$, where \items is the set of all items in \database. This is because \myalgo computes the weight of a transaction directly based on Theorem \ref{theo:vutu}, while \HAISampler needs to compute the local matrix of each item.

\paragraph{Drawing complexity} Algorithm \myalgo draws a pattern of length \uLen by making jumps, where each jump corresponds to selecting an item. Let \trans be the transaction from which the pattern is drawn. The item of the first jump is found in $\complexity(\log(\card{\trans}-(\uLen-1)))$, the second in $\complexity(\log(\card{\trans}-1-(\uLen-2)))$, and so on, until the last in $\complexity(\log(\card{\trans}-(\uLen -1)-(\uLen-\uLen)))$. Therefore, the overall complexity is $\complexity(\sum_{i=0}^{\uLen -1}\log(\card{\trans}-i-(\uLen-1-i)) = \complexity(\sum_{i=0}^{\uLen -1}\log(\card{\trans}-\uLen+1)) = \complexity(\uLen \times \log(\card{\trans}-\uLen+1))$. Considering the complexity to draw the transaction \trans in $\complexity(\log(\card{\database}))$, the complexity to draw a pattern directly from the database \database is $\complexity(\log(\card{\database}) + \maxLen \times \log(\card{\items}-\maxLen+1))$. Then, the complexity to draw $K$ patterns with \myalgo is $\complexity(K \times (\log(\card{\database}) + \maxLen \times \log(\card{\items}-\maxLen+1)))$, which is lower than that of \HAISampler, which is $\complexity(K \times (\log(\card{\database}) + \card{\items}))$.

\paragraph{Storage Complexity} 
Let $\complexity(||\database||)$ represent the storage cost of the entire database in memory before processing (size in memory, Table \ref{tab:benchmark_dataset}). After the preprocessing phase, \myalgo requires $\complexity(\card{\database})$ to store the weight of each transaction of \database. Additionally, the combination values are stored in memory based on Pascal's triangle, utilizing the symmetric property $\Cnk{i}{j} = \Cnk{i}{i-j}$. In this case, the storage cost for the combination values is $\complexity(G(\card{\trans}_{\max}))$ thanks to Property \ref{prop:pascal}, with $\card{\trans}_{\max}$ the length of the longest transaction of \database. Hence, the total complexity of the storage cost of \myalgo is in $\complexity(||\database|| + \card{\database} + G(\card{\trans}_{\max}))$. But, $\card{\database} + G(\card{\trans}_{\max}) <\!\!<\!\!< ||\database||$.

\begin{property}[Storage Cost of Pascal's Triangle]
	\label{prop:pascal}
	Thanks to the symmetric property $\Cnk{i}{j} = \Cnk{i}{i-j}$, the minimal storage cost of a Pascal's triangle with $n+1$ rows is in $\complexity(G(n))$ with:
	\begin{equation}
	G(n)  = \left\{ \begin{array}{cl}
	\left( \frac{n}{2}+1 \right)^2 & ~~~~ \text{if n is even} \\
	(n+1)\times\frac{n+3}{4} & ~~~~ \text{otherwise} 
	\end{array} \right.
	\end{equation}
\end{property}

\begin{proof}(Property \ref{prop:pascal})
    By exploiting the symmetric property $\Cnk{i}{j} = \Cnk{i}{i-j}$, we need only to store the ($\floor{\frac{i}{2}}+1$) first values of row $i$, with $i \in [0..n]$. Let $g(i)=\floor{\frac{i}{2}}+1$ for all $i \in [0..n]$. With $n+1$ rows, the total number of combination values that needed to be stored is compute by $G(n)=\sum_{i=0}^{n}g(i)$. First, let us break down the sum into its individual terms: $g(0)=\floor{\frac{0}{2}}+1=0+1=1$; $g(1)=\floor{\frac{1}{2}}+1=0+1=1$; $g(2)=\floor{\frac{2}{2}}+1=1+1=2$; $g(3)=\floor{\frac{3}{2}}+1=1+1=2$; $g(4)=\floor{\frac{4}{2}}+1=2+1=3$; so on ... Notice that, the term $\floor{\frac{i}{2}}+1$ takes on the values 1, 1, 2, 2, 3, 3, ... as $i$ increases. This is because $g(2k+1)=\floor{\frac{2k+1}{2}}+1=\frac{2k}{2}+\floor{\frac{1}{2}} +1=\frac{2k}{2}+1=\floor{\frac{2k}{2}}+1=g(2k)$, for all positive integer $k$. 
    
    \noindent Case 1: $n$ is even, i.e, $n = 2\times k$ with $k$ a positive integer. In that case, $\floor{\frac{n}{2}}=\floor{\frac{2\times k}{2}}=\frac{n}{2}$. Then, we have $G(n)= 2 \times (\sum_{j=1}^{\frac{n}{2}}j) + g(n) = 2\times \frac{\frac{n}{2}(\frac{n}{2}+1)}{2} + \frac{n}{2}+1=\frac{n}{2}\times(\frac{n}{2}+1)+ \frac{n}{2}+1=(\frac{n}{2}+1)\times(\frac{n}{2}+1)$. Hence, we got $G(n)=(\frac{n}{2}+1)^2$.
    
    \noindent Case 2: $n$ is odd, i.e, $n-1$ is even, and $g(n)=\floor{\frac{n}{2}}+1=\frac{n-1}{2}+1$. Therefore, we have $G(n)= G(n-1)+g(n)=(\frac{n-1}{2}+1)^2+\frac{n-1}{2}+1=(\frac{n+1}{2})^2+\frac{n+1}{2}=\frac{n+1}{2}(\frac{n+1}{2}+1)=\frac{n+1}{2}\times\frac{n+3}{2}$. Hence, we got $G(n)=(n+1)\times\frac{n+3}{4}$.
\end{proof}

\begin{table*}
    \centering
	{\tiny
	\centering
	\caption{Statistics of none-semantic qDB benchmark: \retail, \bmsDeux, \kosarak, \chainstoreFIM and \USCensus from SPMF real databases and IBMGenerator synthetic data \tfiftyItenDtwentyM(utility values are randomly drawn between $10$ and $1,000$). We use the``pympler'' Python package to evaluate the size of these databases before any preprocessing. However, for \tfiftyItenDtwentyM, due to its large volume, we can not obtain its exact size in memory.}
 
	\label{tab:benchmark_dataset}	
	\begin{tabular}{l|rrrrr|rrr}
		Database & \multicolumn{1}{c}{$\card{\database}$} & \multicolumn{1}{c}{$\card{\items}$}  & \multicolumn{1}{c}{$\card{\trans}_{\min}$} & \multicolumn{1}{c}{$\card{\trans}_{\max}$} & \multicolumn{1}{c}{$\card{\trans}_{avg}$} & \multicolumn{1}{|c}{\database size in memory} & \database size on disk & $C^k_n$ in memory (Prop. \ref{prop:pascal}, \maxLen=10)\\
		\hline
		\retail & 88,162 & 16,470 & 1 & 76 & 10.30 & 242.80 MB & 7.36 MB & 30.93 MB\\
		\bmsDeux &  77,512 & 3,340 & 1 & 161 & 1.62 & 99.07 MB & 3.55 MB & 69.82 MB \\
		\kosarak  & 990,002 & 41,270 & 1 & 2,497 & 8.09 & 2.13 GB & 61.39 MB & 1.17 GB \\
		\chainstoreFIM  & 1,112,949 & 46,086 & 1 & 170 & 7.23 & 2.12 GB & 74.48 MB & 73.87 MB \\
		\USCensus  & 1,000,000 & 316 & 48 & 48 & 48.00 & $\approx 10$ GB & 352.60 MB & 18.05 MB\\
		\tfiftyItenDtwentyM  & 20,000,000 & 16,957,575 & 14 & 94 & 52.77 & Out of memory & 7.83 GB & 39.12 MB \\
	\end{tabular}
	}
\end{table*}

\begin{figure*}
	\centering
	
	\includegraphics[width= .9\textwidth]{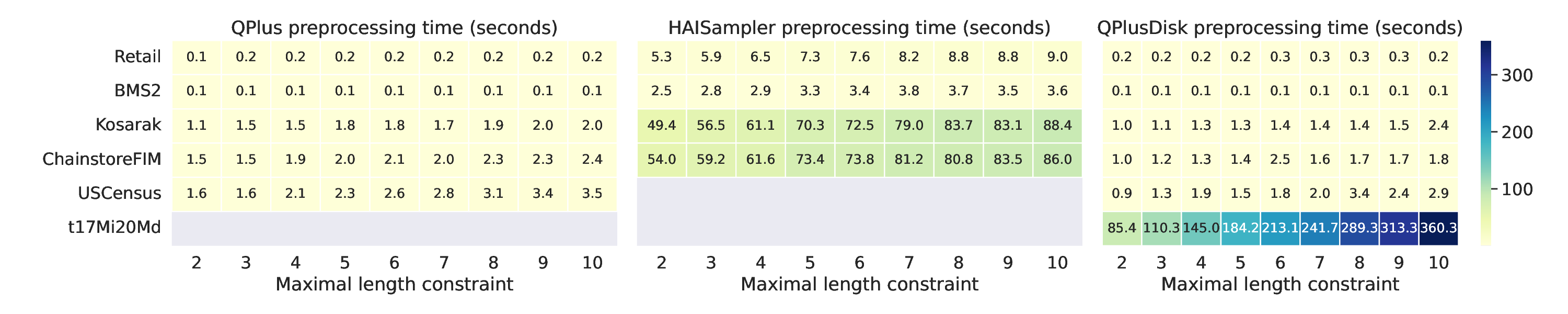}
	\includegraphics[width= .9\textwidth]{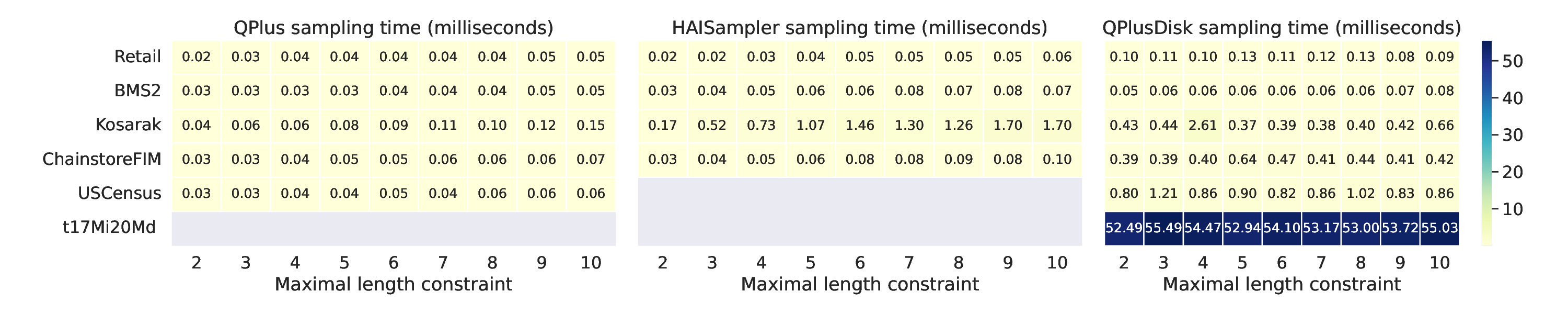}
		
	\caption{Comparing the evolution of execution time based on maximum length constraint (in gray = out of memory)}
\label{fig:preprocessing_under_length_constraint}
\end{figure*}

We recall that our algorithm \myalgo depends on the available memory to perform the dedicated task, which is a problem when we need to explore larger qDBs. To overcome this issue, we propose an on disk approach, denoted by \myalgoOff.
\subsection{\myalgoOff: An on-disk approach of \myalgo}
\label{sub:qplusoff}

The strength of this second method, \myalgoOff, stems from the fact that it avoids storing the entire database in main memory. Only the weights are stored into a list in order to have an exact draw. Algorithm \ref{alg:myalgoOff} formalizes the on-disk version.





\begin{algorithm}[!h]
	{\small
		\caption{\myalgoOff: On-disk of \myalgo 
        }
		\label{alg:myalgoOff}
		\begin{algorithmic}[1]
			\Statex {\bf Input:} A quantitative database $\database$ with an arbitrary total order relation \order, two positive integers \minLen and \maxLen such that $\minLen \leq \maxLen$, and $K$ the number of pattern to draw 
			\Statex {\bf Output:} A list of patterns $[\pattern_1,...,\pattern_K]$, $\pattern_i \sim \utility(\lang_{[\minLen..\maxLen]}(\database))$ 
			
            \Statex //Preprocessing
			\State Let $lisTid$ be the list of transaction identifiers in \database
			\State 
			$W_{[\minLen..\maxLen]}(j)  = \sum_{\uLen = \minLen}^{\maxLen} (\vutu_{\trans_j}(\uLen, \card{{\trans_j}}) \times \utilityLength{[\minLen..\maxLen]}{\uLen})$ for $j \in lisTid$
			\Statex //Drawing: Sampling K patterns independently
			\State Let $selectTid \gets [j_1, \cdots, j_K]$ be a list of $K$ identifiers where 
			$j_k \sim W_{[\minLen..\maxLen]}(\database) = \sum_{j \in lisTid}W_{[\minLen..\maxLen]}(\trans_j) $
			\State Let $Sample$ be an empty list.
			\While{$\card{Sample} < K$ and $\trans_j \in \database$}
			\If{$j \in selectTid$}
			\For{k=1 to count(j, selectTid)} \Comment{An identifier can be drawn multiple times}
			\State Draw a pattern \pattern from $\lang_{[\minLen .. \maxLen]}(\{\trans_j\})$ 
			\State Add \pattern into $Sample$
			\EndFor
			\EndIf
			\EndWhile
			\State \Return{$Sample$} \Comment{A list of $K$ patterns} 
		\end{algorithmic}
	}
\end{algorithm}

\section{Experimental Study}
\label{sec:Experiments}

Table \ref{tab:benchmark_dataset} presents the characteristics of the qDBs used in our study from SPMF\footnote{\url{https://www.philippe-fournier-viger.com/spmf/index.php}} and IBMGenerator\footnote{\url{https://github.com/zakimjz/IBMGenerator}}. Our framework  (Github\footnote{\url{https://github.com/ScalableSamplingInLargeDatabases/QPlus}} and Colab\footnote{\url{https://colab.research.google.com/drive/1yFRH_cIsmfD1OfQ71ruwDELOuQUhyfmZ?usp=sharing}}) is implemented using Python 3.0 in Google Colab with a standard account having 12 GB RAM and Intel(R) Xeon(R) CPU @ 2.20GHz. During all experiments, we set the minimal length constraint to $\minLen=1$, while the maximal length constraint is chosen within the range of $[2..10]$ when applicable. Additionally, we evaluate our approach without any constraint for execution time comparison with \HAISampler\cite{DiopPAKDD2022}, considering its maximal length constraint set to $\infty$. For utility evaluation, we use the average utility function $f_{[\minLen..\maxLen]}(\pattern)=\frac{1}{\card{\pattern}}$ if $\card{\pattern} \in [\minLen..\maxLen]$, and $0$ otherwise. Execution times are repeated 10 times and the standard deviation omitted because too small. Here are the questions that we address:

\begin{enumerate}
	\item Does Theorem \ref{theo:vutu} speed up the preprocessing under length constraint in large qDBs?
	\item Does Theorem \ref{theo:totalweights} allow us to deal with large qDBs without length constraints when the state-of-the-art \HAISampler approach fails?
	\item Is the jumping based on dichotomous random search (Property \ref{prop:dichotomous_random_search}) very interesting in the drawing step?
	\item Are the drawn patterns representative for sub-profiles discovery in knowledge graph profiles?
\end{enumerate}

%
%

\subsection{Analysis of our approaches under length constraint}
\label{sec:ExperimentsAnalysis}

In Figure \ref{fig:preprocessing_under_length_constraint}, we display the execution speed of our algorithms within a length constraint $\maxLen \in [2..10]$. The first line shows the preprocessing time evolution for both \myalgo and \HAISampler. Generally, \myalgo is significantly faster, ranging from at least 15 times to up to 45 times faster than \HAISampler. Despite \myalgoOff being less efficient due to on-disk preprocessing, it remains faster than \HAISampler. Notably, \myalgoOff is the only approach capable of handling the \tfiftyItenDtwentyM, with 2 million transactions stored in a 7.83 GB disk file, in 5 minutes under maximal length constraint of 10 based on Theorem \ref{theo:vutu}.
The second line illustrates the drawing time per pattern under length constraints. \myalgo excels, requiring only 0.14 milliseconds to draw a length-10 pattern in the \kosarak database, whereas \HAISampler takes 1.75 milliseconds for a pattern in the same database. Particularly for databases with long transactions like \kosarak, \myalgo is at least 4 times and at most 16 times faster than \HAISampler. In the third column of the second line, we observe the drawn time per pattern for \myalgoOff. While not as fast as the first approaches, it remains efficient, especially with the larger \tfiftyItenDtwentyM, where it spends around 50 milliseconds to draw a pattern. 
During the sampling phase, \HAISampler exhibits superior performance with small datasets as it does not require computation of any utility sum due to its pre-stored weighting matrix but return Out of memory for larger databases like \USCensus and \tfiftyItenDtwentyM. Conversely, \myalgo and \myalgoOff do not store the UTU. Instead, they employ Theorem \ref{theo:vutu} to determine the acceptance probability of an item while \myalgo remains fast with small ans large databases. These results affirm the effectiveness of our drawing approach, outperforming \HAISampler, thanks to \myalgo jumping based on dichotomous random search. In general, the results in Figure \ref{fig:preprocessing_under_length_constraint} address questions 1 and 3.

\subsection{Analysis of our approaches without length constraint}
\label{sec:ExperimentsSpeedAnalysis}
\begin{table*}
	\centering
 {\small
	\caption{Comparing the preprocessing time (in seconds) with $\maxLen=\infty$}
	\label{tab:preprocessing_without_length_constraint}
	\begin{tabular}{l|rrrrrr}
		Algorithms   & \retail				&		\bmsDeux 			&		\kosarak			&	\chainstoreFIM &  \USCensus& \tfiftyItenDtwentyM \\
		\hline
		\myalgo		&	0.139 $\pm$	 0.032 	&	0.111 $\pm$ 0.042 	&	2.404 $\pm$ 0.604 	&	1.697 $\pm$ 0.364 	&	1.821 $\pm$ 0.624 	& 	$-$ \\
		\HAISampler	&	15.996 $\pm$ 1.133	&	5.509$\pm$ 0.607 	& $-$	& 	$-$ &	$-$	 & 	$-$  \\	
		\myalgoOff	&	0.182 $\pm$ 0.103	&	0.106 $\pm$ 0.130 	& 2.835 $\pm$ 0.149 	& 	1.138 $\pm$  0.132 &	1.562 $\pm$ 0.822	 & 	27.359 $\pm$ 1.268  \\	
	\end{tabular}
 }
\end{table*}

\begin{table*}
	\centering
 {\small
	\caption{Comparing the drawing time per pattern (in milliseconds) with $\maxLen=\infty$}
	\label{tab:drawing_without_length_constraint}
	\begin{tabular}{l|rrrrrr}
		Algorithms   & \retail				&		\bmsDeux 			&		\kosarak			&	\chainstoreFIM	 &  \USCensus & \tfiftyItenDtwentyM  \\
		\hline
		\myalgo		&   0.297 $\pm$	 0.012 	&	1.078 $\pm$ 0.024 	&	120.383 $\pm$ 0.581 	&	1.051 $\pm$ 0.163 	&	0.164 $\pm$ 0.007 		 & 	$-$  \\
		\HAISampler	&	0.190 $\pm$ 0.059	&	0.404 $\pm$ 0.127 	& $-$	& 	$-$ &	$-$ 	 & 	$-$  \\	
		\myalgoOff	&	0.332 $\pm$ 0.053	&	0.590 $\pm$ 0.032 	& 537.338 $\pm$ 0.952 	& 	0.950 $\pm$  0.017 &	0.951 $\pm$ 0.096	 & 	52.813 $\pm$ 0.560  \\
	\end{tabular}
 }
\end{table*}

Tables \ref{tab:preprocessing_without_length_constraint} and \ref{tab:drawing_without_length_constraint} present experimental results on execution times (preprocessing and drawing) for our approaches and a baseline without length constraint.

Table \ref{tab:preprocessing_without_length_constraint} shows that our preprocessing approaches remain very efficient without length constraint even with large databases while \HAISampler can not work with these large databases. With \retail and \bmsDeux, \myalgo spends $0.139$ seconds and $0.111$ second with a small standard deviation while \HAISampler needs $5.996$ seconds and $5.509$ respectively with a larger standard deviation. A particular interesting remark is that \myalgoOff is too very fast without length constraint. It is as good as \myalgo while being more faster than \HAISampler. These results consolidate the effectiveness of Theorem \ref{theo:totalweights} for processing large databases.
Table \ref{tab:drawing_without_length_constraint} shows the drawing times per pattern for the three approaches. \HAISampler is more faster in qDBs that contain not too long transactions like \retail and \bmsDeux. We can see that without constraint, all pattern are drawn from the larger transactions when dealing with databases that contain very long transactions like \kosarak (2,497 in one transaction). In this database, \myalgo spends at average $120.383$ milliseconds to draw a pattern, because all of them are too long. \myalgoOff needs in average $52.813$ milliseconds to draw  pattern without length constraint from the largest database \tfiftyItenDtwentyM. The results indicate that the jumping based on dichotomous random search is particularly beneficial for sampling patterns when the maximal constraint tends to $\infty$, especially in databases containing very long transactions. 
These experiments affirm the effectiveness of Theorem \ref{theo:vutu}, Theorem \ref{theo:totalweights}, and Property \ref{prop:dichotomous_random_search} in efficiently sampling patterns in large qDBs without length constraints, then supporting responses to questions 2 and 3.

%

\subsection{HUP-based sub-profiles in knowledge graph}
\label{sec:sub-profile}

Our use case is dedicated for diverse representative sub-profile discovery in knowledge graph profiling \cite{Principe-ABSTAT,DiopJOCCH23}. The diversity is achieved thanks to the random nature while the representativeness is guaranteed by the exactitude of our approach (Property \ref{prop:soundness}). 
The profiles \DOREMUS, \BENICULTURALI and \DBpedia in Table \ref{tab:dataset} present enormous difficulties for visualization with the \KGSUMVIZ\footnote{\url{https://github.com/DTTProfiler/DTTProfiler}} tool, even if they are not really so big. However, with our approach, we can allow the user to choose the maximum size of the sub-profiles (i.e, number of triples) to be extracted from the initial profile which is equal to $\maxLen\times N$ where \maxLen is the maximum length constraint and $N$ is the number of desired sampled patterns.

\begin{table}[htp!]
\centering
{\small
\caption{Semantic qDB of Knowledge graph profiles}	
	\label{tab:dataset}	
	\begin{tabular}{l|r|r|r}
		& \multicolumn{3}{c}{statistics for TT profile} \\
		\ABox & nb. of TT & nb. of links  & nb. of nodes \\
		\hline
		\DOREMUS  &  2,399	&	1,785	&	115 \\
		\BENICULTURALI  &  641&	7,518	&	440   \\
		\DBpedia  &  404	&	6,010	&	204	 \\
	\end{tabular}
 }
\end{table}

\paragraph{Qualitative evaluation of sampled pattern-based sub-profile} We begin by demonstrating the representativeness of patterns sampled by \myalgo across the entire qDB, comparing it with the \Bootstrap sampling approach. \Bootstrap randomly selects a transaction and returns a pattern of length between \minLen and \maxLen. The distribution of distinct sampled patterns is shown in Figure \ref{fig:comparison_quality_QPLUSvsBootstrap}, comparing \myalgo with \Bootstrap, considering a maximum length $\maxLen=5$. The patterns drawn by \myalgo prove more representative, with higher average utility and greater abundance in the sample. This affirms the effectiveness of \myalgo in constructing representative sub-profiles within large knowledge graph profiles.

\begin{figure}
	\centering
	\includegraphics[width=.16 \textwidth]{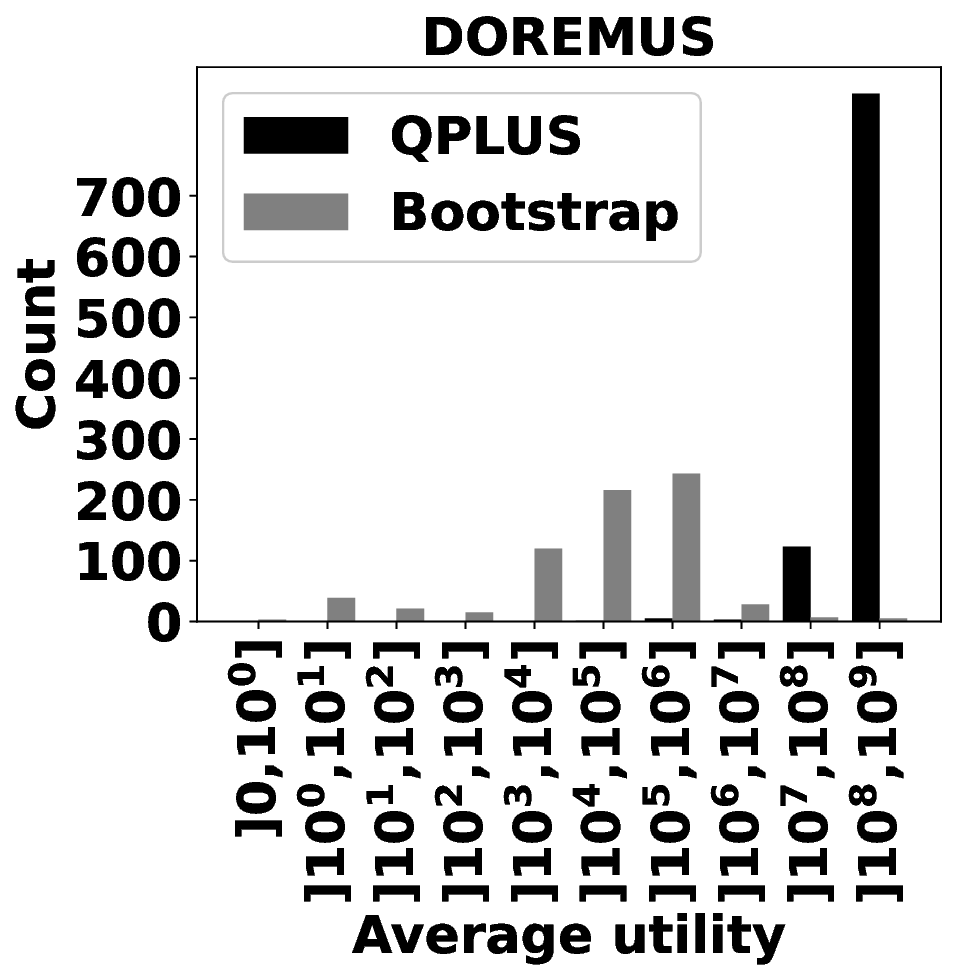}
	\includegraphics[width=.13 \textwidth]{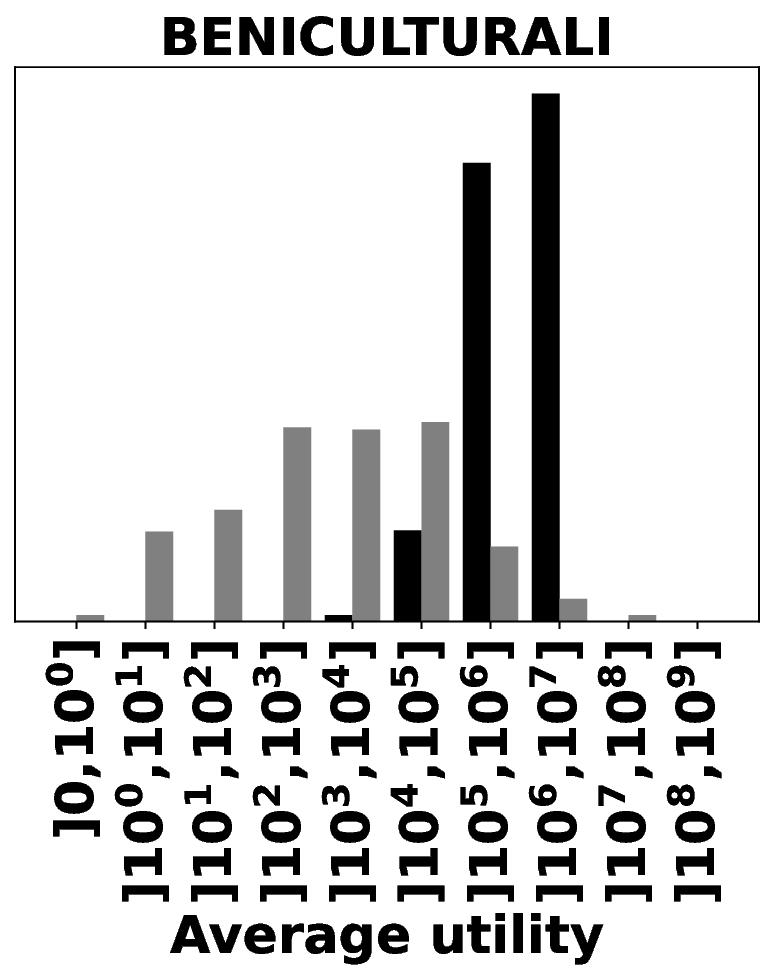}
	\includegraphics[width=.13 \textwidth]{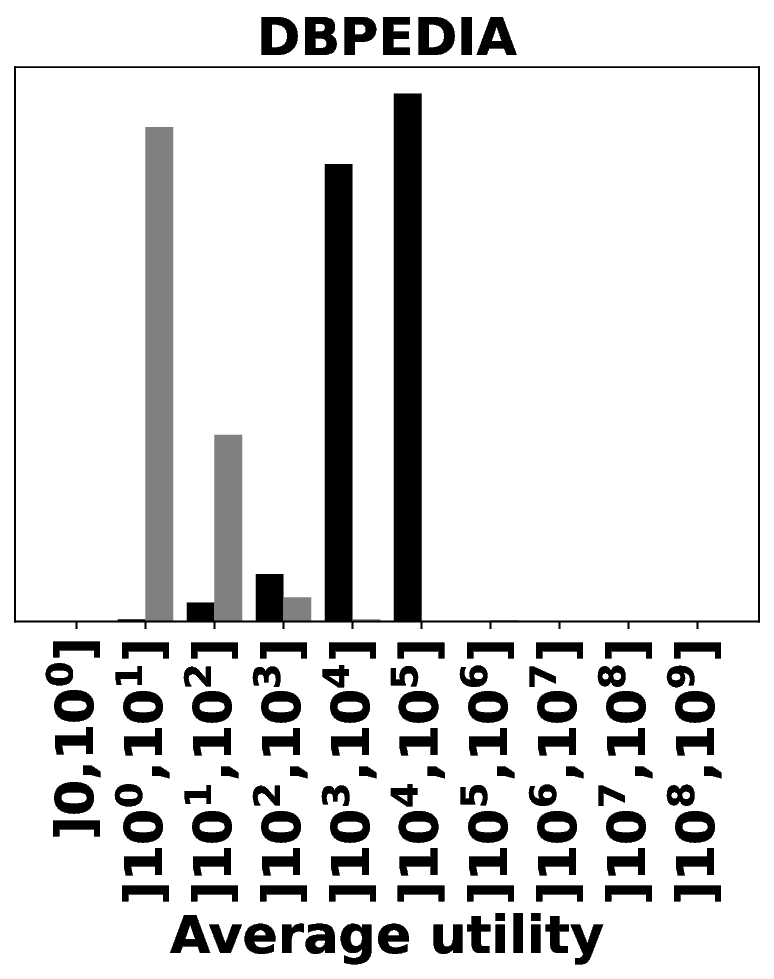}
	\caption{Representativeness of 1,000 drawn patterns by \myalgo and Bootstrap}
	\label{fig:comparison_quality_QPLUSvsBootstrap}
\end{figure}

To highlight the added value, we present a violin plot in Figure \ref{fig:plus-values}, depicting the distribution of means and their 95\% confidence intervals for both \myalgo and \Bootstrap. Each pattern is weighted by its average utility, normalized by the sum of pattern utilities in the union of samples from \myalgo and \Bootstrap. The plot reveals that \Bootstrap exhibits higher variance, indicating lower precision in estimating the distribution of pattern utilities. In contrast, \myalgo, with a centered white dot representing the median, offers a more representative view of the pattern space. These consistent outcomes underscore the added value of \myalgo over \Bootstrap across these knowledge graph profiles.

\begin{figure}
	\centering
	\includegraphics[width=.15 \textwidth]{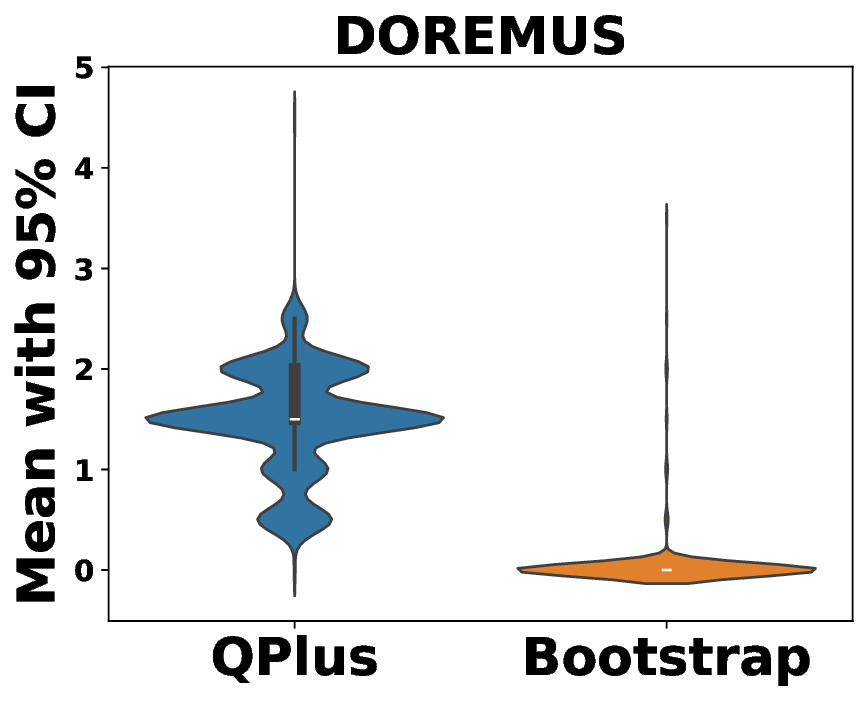}
	\includegraphics[width=.14 \textwidth]{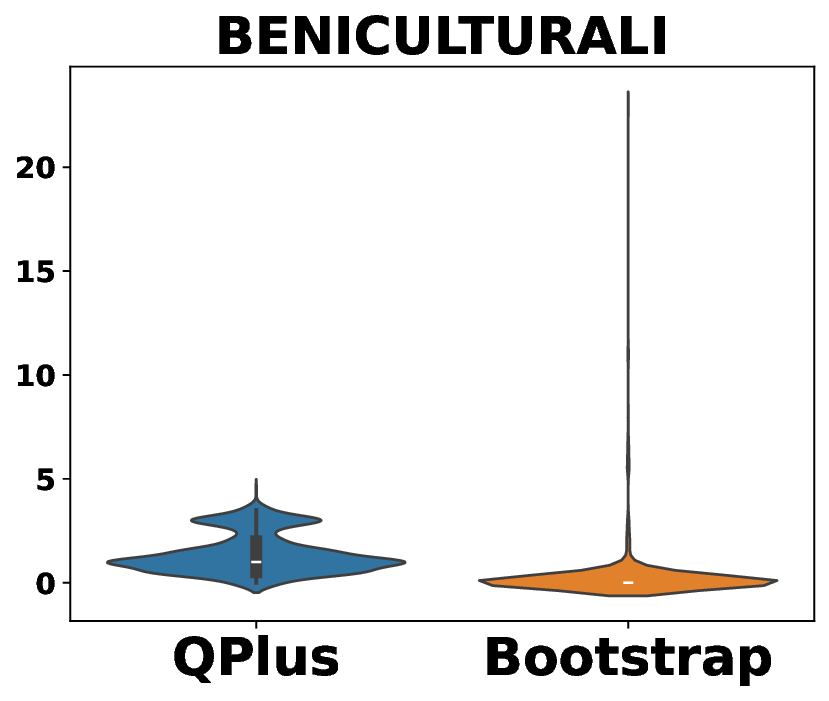}
	\includegraphics[width=.14 \textwidth]{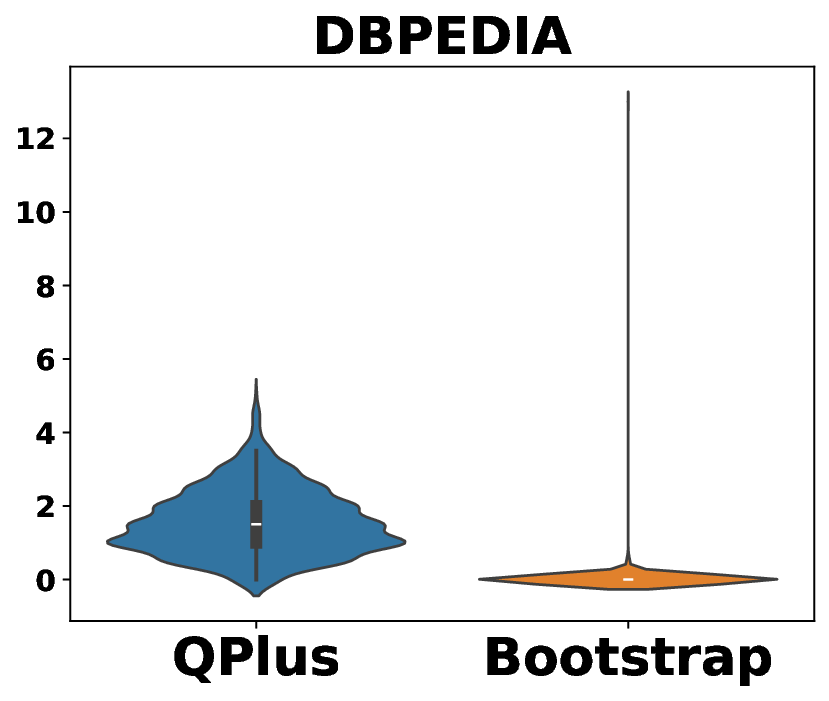}
	\caption{Violin plots with mean and 95\% confidence intervals for the plus-value of \myalgo on Bootstrap for 1,000 patterns}
	\label{fig:plus-values}
\end{figure}

Finally, Figure \ref{fig:viz-subprofile} displays sub-profiles derived from patterns drawn by \myalgo where nodes represent collections of concepts or terms. The stemming sub-profiles show the structure of the original knowledge graph profile. For example, \DOREMUS and \DBpedia profiles concentrate around single nodes (`prov\#Entity') and `Company' respectively, resulting in ego-network-like arrangements. On the other hand, \BENICULTURALI profile information spans multiple nodes, contributing to less interconnected nodes in the stemmed sub-profile graph. No node is isolated; each is represented as a triple containing a predicate linking a set of concepts as the subject to a set of concepts or terms as the object.

\begin{figure}
	\centering
    \begin{tabular}{c|c|c}
        \begin{minipage}{.15 \textwidth}
            \includegraphics[width= \linewidth]{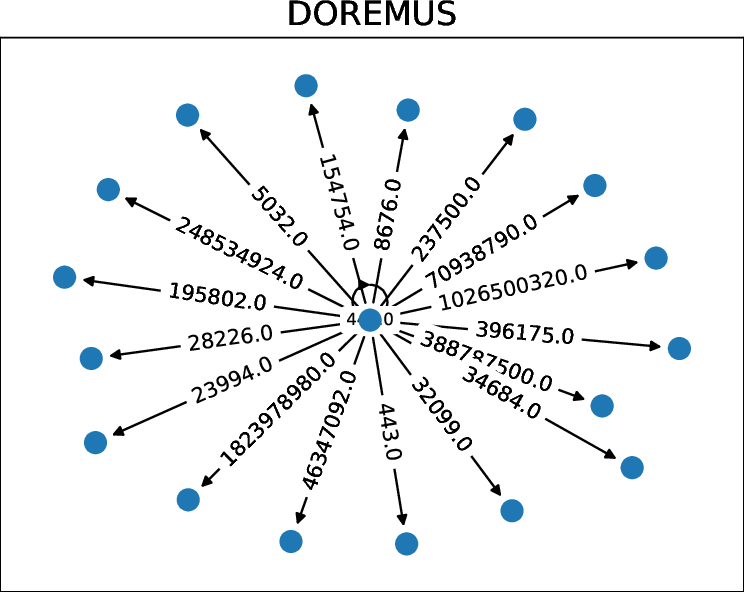}
        \end{minipage}
        &
        \begin{minipage}{.14 \textwidth}
            \includegraphics[width=\linewidth]{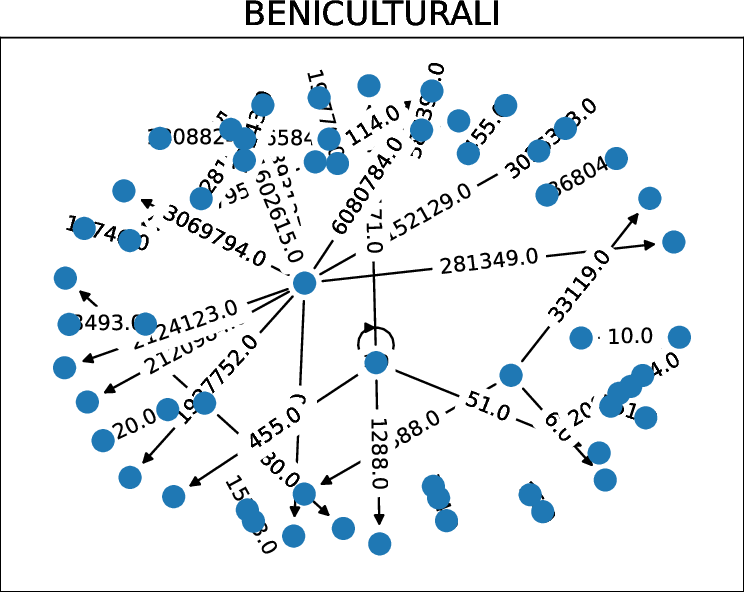}
            \label{fig:viz-subprofile1}
        \end{minipage}
        &
        \begin{minipage}{.14 \textwidth}
            \includegraphics[width=\linewidth]{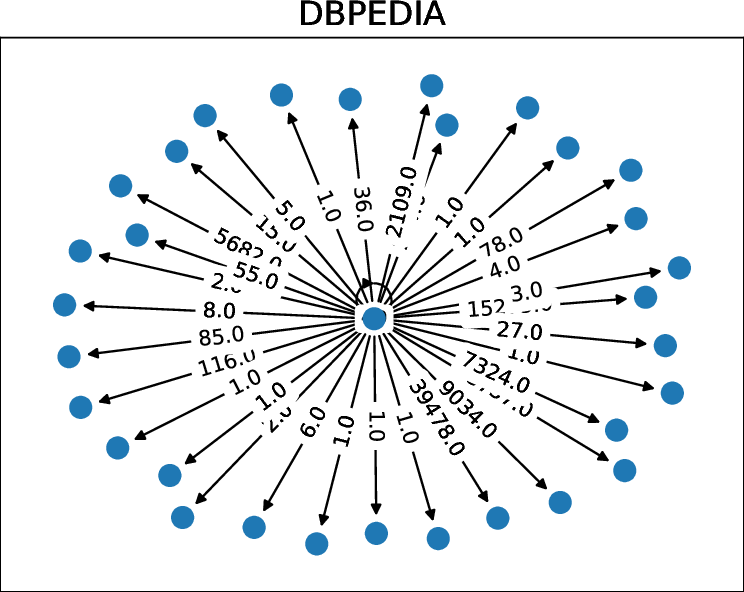}
            \label{fig:viz-subprofile1}
        \end{minipage}
        
        \\
        \begin{minipage}{.14 \textwidth}
            \includegraphics[width= \linewidth]{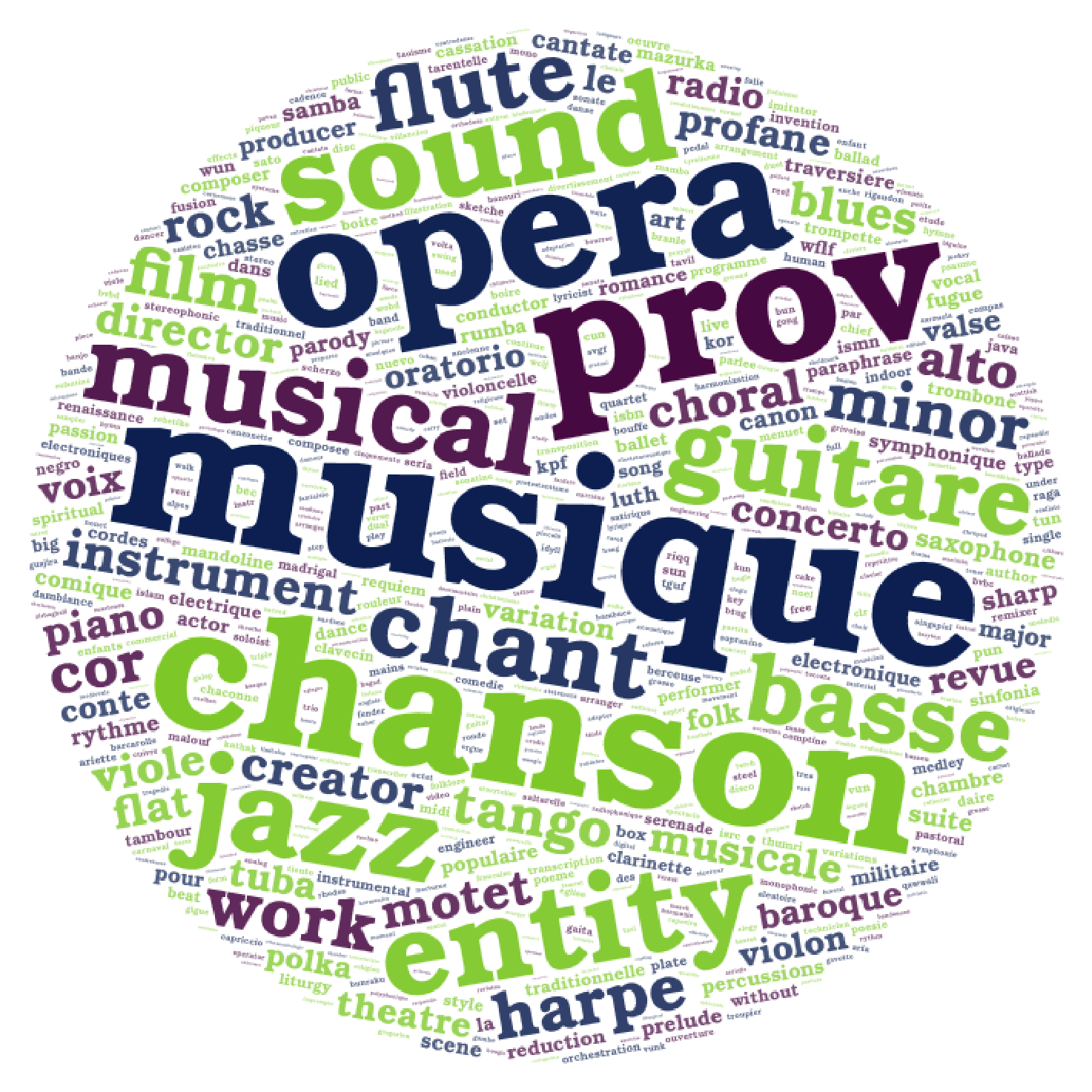}
        \end{minipage}
        &
        \begin{minipage}{.15 \textwidth}
            \includegraphics[width=\linewidth]{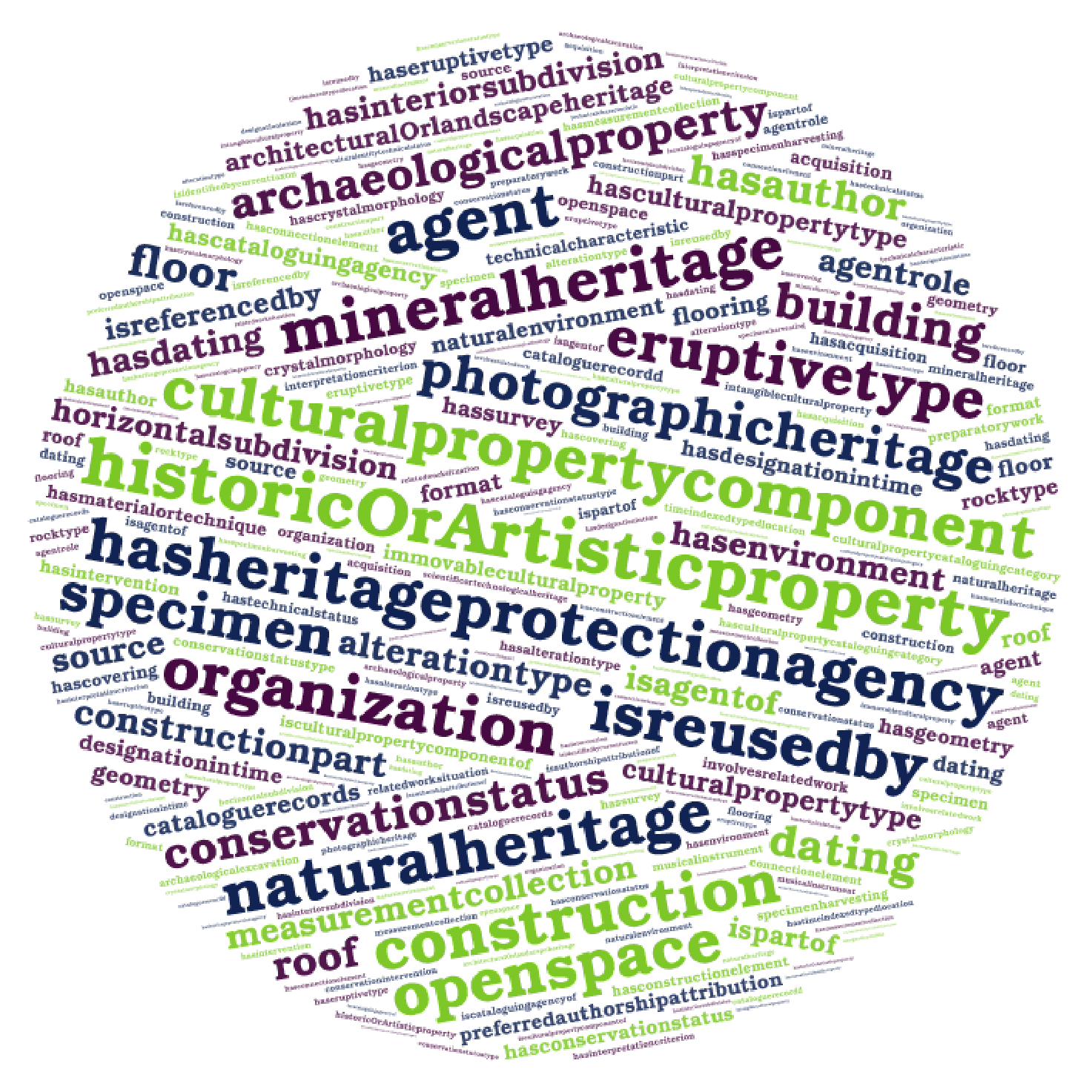}
            \label{fig:viz-subprofile2}
        \end{minipage}
        &
        \begin{minipage}{.14 \textwidth}
            \includegraphics[width=\linewidth]{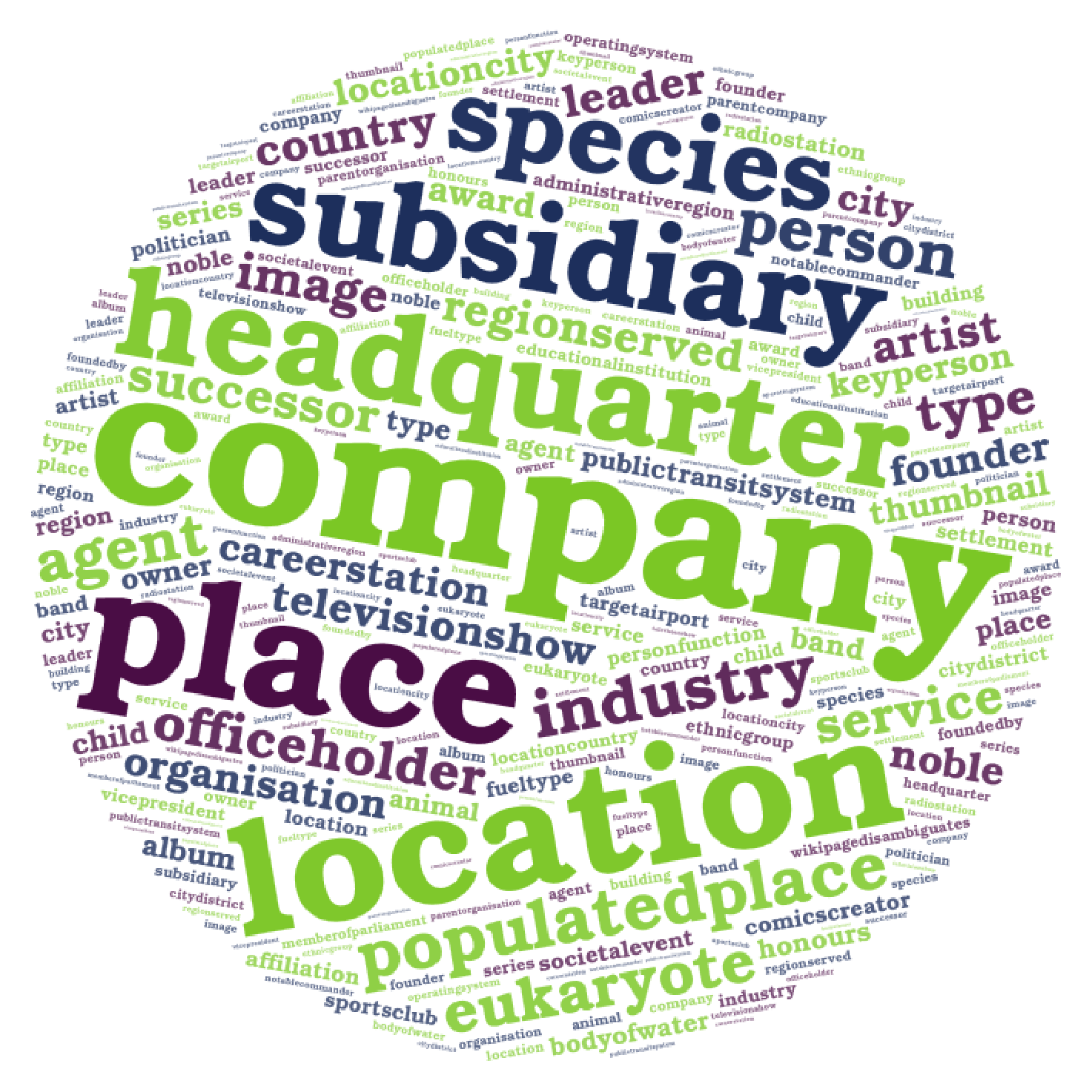}
            \label{fig:viz-subprofile2}
        \end{minipage}
    \end{tabular}

	\caption{Visualizing sub-profiles from $N=10$ patterns of maximal length $\maxLen=5$ drawn by \myalgo}
	\label{fig:viz-subprofile}
\end{figure}

\subsection{Response time for interactive sub-profile discovery} To integrate \myalgo into the framework for user-centric sub-profile discovery, we need fast processing and fast drawing capabilities. This response time depends on the characteristics of the corresponding qDB, particularly the longest transaction, the average transaction size, and the number of transactions. With $\maxLen=5$, the processing times for \DOREMUS, \BENICULTURALI, and \DBpedia are $0.004$ seconds, $0.018$ seconds, and $0.023$ seconds respectively, while the drawing times per pattern are $0.069$ milliseconds, $0.127$ milliseconds, and $0.141$ milliseconds respectively. Hence, it is possible to return a sub-profile with hundreds of nodes and thousands of links in less than $1$ millisecond. It is important to note that in interactive systems, even a small difference between $1$ and $15$ seconds holds significant importance for users. This is particularly true as they engage in multiple iterations over large databases to visualize sub-profiles interactively, which motivates our work. But, for clear visualization in practice, the value of $\maxLen \times N$ should not be too high.
\section{Conclusion}
\label{sec:Conclusion}

This paper introduces a groundbreaking approach—the first of its kind—for extracting sub-profiles from knowledge graph profiles, employing high utility patterns and an output sampling method. With the presentation of two original theorems, we achieve an efficient computation of the total sum of utility patterns for a given quantitative transaction. The concept of Upper Triangle Utility (UTU), embodied by a non-materialized upper triangle matrix, emerges as a pivotal practical element for deducing transaction weights. Building upon the theorems and properties derived from UTU, we unveil effective two-step algorithms tailored for the exploration of vast quantitative databases. Our paper makes significant contributions by providing a robust framework for exploratory data analysis in weighted data, delivering innovative theorems and algorithms for handling large-scale quantitative databases.

{\small
\bibliographystyle{IEEEtran}
\bibliography{QPlus}
}

\end{document}